\documentclass[12pt,onecolumn]{IEEEtran}

\usepackage{cite}
\usepackage{pslatex,bm}
\usepackage{epsfig}
\usepackage{amsmath}
\usepackage{amssymb}
\usepackage{array}
\usepackage{multicol}
\usepackage{color}
\usepackage{url}
\usepackage{dsfont}
\psfull

\newtheorem{lemma}{Lemma}

\newtheorem{definition}{Definition}
\newtheorem{coro}{Corollary}

\DeclareMathAlphabet{\mathpzc}{OT1}{pzc}{m}{it}

\renewcommand{\baselinestretch}{2}

\def\calW{{\mathcal W}}

\def\ba{{\mathbf a}}
\def\bb{{\mathbf b}}

\def\be{{\mathbf e}}

\def\bA{{\mathbf A}}
\def\bB{{\mathbf B}}

\def\bI{{\mathbf I}}
\def\bK{{\mathbf K}}

\def\bT{{\mathbf T}}

\def\diag{{\rm diag}}
\def\tr{{\rm tr}}

\def\b0{{\mathbf 0}}

\begin{document}

{\title{\vspace{-1cm}On the Secrecy Rate Region of a Fading \\\vspace{-0.7cm}Multiple-Antenna Gaussian Broadcast Channel with\\\vspace{-0.7cm} Confidential Messages and Partial CSIT}
}

\author{Pin-Hun Lin\hspace{1.8cm} Chien-Li Su\hspace{1.8cm} Hsuan-Jung Su\\

\normalsize{\it pinhsunlin@gmail.com} \hspace{0.3cm} {\it theone21cent@gmail.com\mbox{    }     }
\vspace{0.5cm}{\it hjsu@cc.ee.ntu.edu.tw}\\
Graduate Institute of Communications Engineering, Department of Electrical Engineering, \\
National Taiwan University, Taipei, Taiwan 10617\\
}
\maketitle \thispagestyle{empty} \vspace{-15mm}
{\renewcommand{\baselinestretch}{2}

\begin{abstract}
In wiretap channels the eavesdropper's channel state information
(CSI) is commonly assumed to be known at transmitter, fully or
partially. However, under perfect secrecy constraint the
eavesdropper may not be motivated to feedback any correct CSI. In
this paper we consider a more feasible problem for the transmitter
to have eavesdropper's CSI. That is, the fast fading
multiple-antenna Gaussian broadcast channels (FMGBC-CM) with
confidential messages, where both receivers are legitimate users
such that they both are willing to feedback accurate CSI to
 maintain their secure transmission, and not to be eavesdropped
 by the other. We assume that only the statistics of the channel
 state information are known by the transmitter. We first show the
  necessary condition for the FMGBC-CM not to be degraded to the
  common wiretap channels.
 Then we derive the achievable rate region for the FMGBC-CM where the channel input covariance matrices and the
    inflation factor are left unknown and to be solved. After
 that we provide an analytical solution to the channel input covariance matrices. We also propose an iterative algorithm to
  solve the channel input covariance matrices and the inflation
  factor. Due to the
    complicated rate region formulae in normal SNR, we resort to low SNR analysis to investigate
    the characteristics of the channel.
   Finally, numerical examples show that under perfect secrecy constraint
   both users can achieve positive rates simultaneously, which verifies
   our derived necessary condition. Numerical results also elucidate the
   effectiveness of the analytical solution and proposed algorithm of the channel input covariance matrices and the inflation factor under different conditions.
\end{abstract}

In wiretap channels the eavesdropper's channel state information
(CSI) is commonly assumed to be known at transmitter, fully or
partially. However, under perfect secrecy constraint the
eavesdropper may not be motivated to feedback any correct CSI. In
this paper we consider a more feasible problem for the transmitter
to have eavesdropper's CSI. That is, the fast fading
multiple-antenna Gaussian broadcast channels (FMGBC-CM) with
confidential messages, where both receivers are legitimate users
such that they both are willing to feedback accurate CSI to
 maintain their secure transmission, and not to be eavesdropped
 by the other. We assume that only the statistics of the channel
 state information are known by the transmitter. We first show the
  necessary condition for the FMGBC-CM not to be degraded to the
  common wiretap channels.
 Then we derive the achievable rate region for the FMGBC-CM. After
 that we propose an iterative algorithm to
  solve the channel input covariance matrices and the inflation factor.
   Finally, numerical examples show that under perfect secrecy constraint
   both users can achieve positive rates simultaneously, which verifies
   our derived necessary condition. Numerical results also elucidate the
   effectiveness of the proposed algorithm in convergence speed.

\section{ Introduction}
Traditionally, the security of data transmission has been ensured
by the key-based enciphering. However, for secure communication in
large-scale wireless networks, the key distributions and
managements may be challenging tasks
\cite{Liang_same_marginal}\cite{Secureconnect}. The physical-layer
security introduced in
\cite{Wyner_wiretap}\cite{csiszar1978broadcast} is appealing due
to its keyless nature. One of the fundamental setting for the
physical-layer security is the wiretap channel. In this channel, the transmitter wishes to send messages securely to a legitimate
receiver and to keep the eavesdropper as ignorant of the message
as possible. Wyner first characterized the secrecy capacity of the
discrete memoryless wiretap channel \cite{Wyner_wiretap}. The
secrecy capacity is the largest rate communicated between the
source and legitimate receiver with the eavesdropper knowing no
information of the messages. Motivated by the demand of high data
rate transmission, the multiple antenna systems with security
concern were considered by several works. In \cite{Shafiee_secrecy_2_2_1_J}, Shafiee and Ulukus
first proved the secrecy capacity of a Gaussian channel with two-input, two-output,
single-antenna-eavesdropper. Then the authors of
\cite{Khisti_MIMOME,Oggier_MIMOME, Liu_MIMO_wiretap} extended the
secrecy capacity result to the Gaussian multiple-input
multiple-output, multiple-antenna-eavesdropper channel. On
the other hand, the impacts of fading channels on the secure transmission were
considered in \cite{Liang_fading_secrecy}. Note that \cite{Shafiee_secrecy_2_2_1_J,Khisti_MIMOME,Oggier_MIMOME,
Liu_MIMO_wiretap,Liang_fading_secrecy} require full channel state
information at the transmitter (CSIT). When there is only partial CSIT,
several works considered the secure transmission under this condition \cite{Khisti_MISOME,Goel_AN,Pulu_ergodic,Li_fading_secrecy_j,
gopala2008secrecy}. The artificial noise (AN) assisted secure
beamforming is a promising technique for the partial CSIT cases, where
in addition to the message-bearing signal, an AN is intentionally
transmitted to disrupt the eavesdropper's reception
\cite{Khisti_MISOME}\cite{Goel_AN}.  Indeed, adding AN in transmission is
crucial in increasing the secrecy rate in fading wiretap channels. However, the covariance
matrices of AN in \cite{Goel_AN}\cite{Khisti_MISOME} is
heuristically selected without optimization, and the resulting
secrecy rate is not optimal. In \cite{Pulu_ergodic}, the secure transmission under fast fading channels with only statistical CSIT and without AN is considered.
Although the secrecy capacity for single antenna system with
partial CSIT was found in \cite{gopala2008secrecy}, the decoding
latency of the transmission scheme proposed in
\cite{gopala2008secrecy} is much longer than the common fast fading channels, e.g., \cite{Khisti_MISOME,Goel_AN,Pulu_ergodic,Li_fading_secrecy_j}, and may be unacceptable in practice.

However, the assumptions of wiretap channels with full or partial CSIT may not be practical. That is, the eavesdroppers needs to feedback the perfect/statistical CSI to transmitter or the transmitter needs to know this CSI by some means. On the contrary, the eavesdroppers may not be motivated to feedback this information. Furthermore, the eavesdroppers may feedback the wrong CSI to destroy the secure transmission. Thus in this paper, we consider the multiple antenna Gaussian broadcast channel with confidential messages (MGBC-CM) \cite{Liu_MISO} under fast fading channels (abbreviated as \textit{FMGBC-CM}). In the FMGBC-CM, both receivers are legitimate users such that they both are willing to feedback accurate CSI to maintain their secure transmission, and not to be eavesdropped by the other user. In the considered FMGBC-CM, we
assume that the transmitter only has the statistics of the channels from both receivers. This is to taking the practical issues into account, such as the limited bandwidth of the feedback channels or the speed of the channel estimation at the receivers. And to the best knowledge of the authors, this problem has not been considered in the literature.

The main contribution of this paper is to provide an achievable rate region with explicit channel input covariance matrices of both users. An iterative algorithm is proposed to solve the inflation factor of the linear assignment Gel'fand-Pinsker coding (LA-GPC) \cite{Gelfand_Pinsker} used in the adopted transmission scheme. To accomplish these, we first classify the non-trivial cases such that the FMGBC-CM is not degraded as the conventional wiretap channel, i.e., both users have positive secure transmission rates, by the same marginal property. We then prove that the MISO GBC-CM with identical and independently distributed (i.i.d.) Rayleigh fading channels is degraded as the MISO Gaussian wiretap channel. Thus in this paper we consider the non-i.i.d. Rayleigh fading MISO and Rician fading MISO BC-CM, respectively.

The rest of the paper is organized as follows. In Section
\ref{Sec_system_model} we introduce the considered system model. We then provide the necessary conditions for the FMGBC-CM to be not degraded as a conventional wiretap channel in Section \ref{sec_necessary_condition}.
In Section \ref{Sec_rate_region_MGBC_CM}, we derive the achievable secrecy rate region of the FMGBC-CM. An achievable selection of the channel input covariance matrices and an iterative scheme for solving the inflation factor are also provided to calculate the explicit rate region. In Section \ref{Sec_low_SNR}, we demonstrate the secrecy rate region in low SNR regime and the optimal signaling. In Section \ref{Sec_simulation} we illustrate the
numerical results. Finally, Section \ref{Sec_conclusion}
concludes this paper.

\section{System model}\label{Sec_system_model}

In this paper we consider the FMGBC-CM system as shown in Fig. \ref{Fig_sys},
where the transmitter has $n_T$ antennas and the receiver 1 and 2 each has single antenna. The received
signals at the $j$th receivers, $j=$ 1 or 2, can be respectively represented as \footnote{In this paper, lower and upper case bold alphabets denote vectors
and matrices, respectively. Italic upper alphabets with and without boldface denote random vectors and variables, respectively. The
$i$th element of vector $\ba$ is denoted by $a_i$. The
superscript $(.)^H$ denotes the transpose complex conjugate. The
superscript $a^c$ denotes the complement of $a$.
$|\mathbf{A}|$ and $|a|$ represent the determinant of the square
matrix $\mathbf{A}$ and the absolute value of the scalar variable
$a$, respectively. A diagonal matrix whose diagonal entries are
$a_1 \ldots a_k$ is denoted by $diag(a_1 \ldots a_k)$. The trace
of $\mathbf{A}$ is denoted by $\tr(\mathbf{A})$. We define $C(x)
\triangleq \log(1+x)$ and $(x)^+\triangleq\max\{0,\,x\}$. The
mutual information between two random variables is denoted by
$I(;)$. $\mathbf{I}_n$ denotes the $n$ by $n$ identity matrix.
$\mathbf{A}\succ 0$ and $\mathbf{A}\succeq 0$ denote that
$\mathbf{A}$ is a positive definite and positive semi-definite
matrix, respectively.}
\begin{align}
Y_{j,k}={\bm{H}_{j,k}}^H\bm{X}_k+N_{j,k}, \label{EQ_Eve}
\end{align}
where $\bm{X}_k\in\mathds{C}^{n_T\times 1}$ is the transmit vector,
$k$ is the time index, respectively, $\bm{H}_{j,k}$ denotes the
fading vector channels from the transmitter to the $j$th receiver,
and $N_{j,k}$ is the circularly symmetric complex additive white
Gaussian noise with variances one at receiver $j$, respectively. In
this system, we assume that only the statistics of both channels are
known at transmitter, to take the practical system design issues
into account, such as the limited bandwidth of the feedback channels
or the speed of the channel estimation at the receivers. We also
assume that { both the receivers perfectly know all channel
vectors}. Without loss of generality, in the following we omit the
time index to simplify the notation. We consider the power
constraint as
\[
\mbox{tr}(E[\bm{X}\bm{X}^H])\leq P_T.
\]

The perfect secrecy and secrecy capacity are defined as follows.
Consider a $(2^{nR_1}, 2^{nR_2},n)$-code with an encoder that maps
the message $W_1\in \calW_1=\{1,2,\ldots, 2^{nR_1}\}$ and
$W_2\in\calW_2=\{1,2,\ldots, 2^{nR_2}\}$ into a length-$n$ codeword,
and receiver 1 and receiver 2 map the received sequence $Y_1^n$ and
$Y_2^n$ (the collections of $Y_1$ and $Y_2$, respectively, over code
length $n$) from the MISO channel to the estimated message $\hat
W_1\in\calW_1$ and $\hat W_2\in\calW_2$, respectively. Since
$\bm{H}_1$ and $\bm{H}_2$ are both known at receiver 1 and 2,
respectively, we can treat them as the channel outputs similar to
\cite{caire_channel_with_SI}. We then have the following definition
of secrecy capacity region.

\begin{definition}[Secrecy capacity region] \label{Def_Perfect}
 {\it Perfect secrecy with rate pair $(R_1,\,R_2)$
is achievable if, for any positive $\varepsilon$ and $\varepsilon'$, there
exists a sequence of $(2^{nR_1},2^{nR_2}, n)$-codes and an integer $n_0$ such
that for any $n>n_0$
\begin{align}\label{eq_equivocation_given_h}
&I(W_j;Y_{j^c}^n,\bm{H}_{j^c}^n)/n<\varepsilon, \mathrm{and
}\;\;{\rm Pr}(\hat{W}_j\neq W_j) \leq \varepsilon',
\end{align}
where $\bm{H}_1^n$ and $\bm{H}_2^n$ are the collections of $\bm{H}_1$ and
$\bm{H}_2$ over code length $n$, respectively. The {\rm secrecy capacity region} is the closure of the set of all achievable
rate pairs $(R_1,\,R_2)$.}
\end{definition}

Note that as shown in the footnote, italic upper alphabets with and without boldface denote random vectors and variables, respectively. By treating $\bm{H}_1$ and $\bm{H}_2$ as the channel outputs, we can extend the achievable rate region of the discrete memoryless MBC-CM from \cite{Liu_MISO} as
\[
(R_1,R_2)\in \mbox{co}\left\{\underset{\varpi\in\Omega}\bigcup \mathcal{R}_I(\varpi)\right\},
\]
where co$\{.\}$ denotes the convex closure;
$\mathcal{R}_{I}(\varpi)$ denotes the union of all $(R_1,R_2)$
satisfying { \begin{align} R_j &\leq
(I(\bm{V}_j;Y_j,\bm{H}_j,\bm{H}_j^c)-I(\bm{V}_j;Y_j^c,\bm{H}_j,\bm{H}_j^c,\bm{V}_j^c))^+,\label{EQ_R2_DMC}
\end{align}}
for any given joint probability density $\varpi$ belonging to the class of joint probability densities\\
 $p(\bm{v}_1,\bm{v}_2,\bm{x},y_1,y_2,\bm{h}_1,\bm{h}_2)$, denoted by  $\Omega$, that factor as $p(\bm{v}_1,\bm{v}_2)p(\bm{x}|\bm{v}_1,\bm{v}_2)p(y_1,y_2,\bm{h}_1,\bm{h}_2|\bm{x})$; $\bm{V}_1$ and $\bm{V}_2$ are the auxiliary random vectors for user 1 and 2, respectively.\\
Note that we can further rearrange the right hand side (RHS) of
\eqref{EQ_R2_DMC} with $j=1$ as {
\begin{align}
R_1 \overset{(a)}\leq &(I(\bm{V}_1;Y_1,\bm{H}_1,\bm{H}_2)-I(\bm{V}_1;\bm{V}_2,\bm{H}_1,\bm{H}_2)-I(\bm{V}_1;Y_2|\bm{V}_2,\bm{H}_1,\bm{H}_2))^+\notag\\
\overset{(b)}=&(I(\bm{V}_1;Y_1,\bm{H}_1,\bm{H}_2)-I(\bm{V}_1;\bm{V}_2)-I(\bm{V}_1;Y_2|\bm{V}_2,\bm{H}_1,\bm{H}_2))^+\notag\\
\overset{(c)}=&(I(\bm{V}_1;Y_1|\bm{H}_1,\bm{H}_2)+I(\bm{V}_1;\bm{H}_1,\bm{H}_2)-I(\bm{V}_1;Y_2|\bm{V}_2,\bm{H}_1,\bm{H}_2)-I(\bm{V}_1;\bm{V}_2))^+\notag\\
\overset{(d)}=&(I(\bm{V}_1;Y_1|\bm{H}_1,\bm{H}_2)-I(\bm{V}_1;Y_2|\bm{V}_2,\bm{H}_1,\bm{H}_2)-I(\bm{V}_1;\bm{V}_2))^+,\label{EQ_R1_DMC2}
\end{align}
} where (a) is by applying the chain rule of mutual information to
the second term on the RHS of \eqref{EQ_R2_DMC}; (b) is due to
$\bm{V}_1$ and $\bm{V}_2$ are independent of $\bm{H}_2$; (c) is
again applying the chain rule to the first term; (d) is due to the
fact that since there is only statistical CSIT, $\bm{V}_1$ is
independent of $\bm{H}_1$. Thus $I(\bm{V}_1;\bm{H}_1)=0$. Similarly,
we can rearrange $R_{2}$ as
\begin{equation}
R_2 \leq
(I(\bm{V}_2;Y_2|\bm{H}_1,\bm{H}_2)-I(\bm{V}_2;Y_1|\bm{V}_1,\bm{H}_1,\bm{H}_2)-I(\bm{V}_1;\bm{V}_2))^+.\label{EQ_R2_DMC2}
\end{equation}

\section{Conditions for non-degraded FMGBC-CM}\label{sec_necessary_condition}
{ In this section, we explain our first main result of this paper
which helps to characterize the secrecy rate region, i.e., the
necessary conditions to exclude the cases that one of the two
receivers of FMGBC-CM always has zero rate. Note that if such cases
happens, the FMGBC-CM degrades as a normal fast fading wiretap
channel. And the achievable rates or capacities of fast fading
wiretap channels with different degree of knowledge of CSIT are
widely discussed, e.g.,
\cite{Goel_AN}\cite{Li_fading_secrecy_j}\cite{Pulu_ergodic}\cite{gopala2008secrecy}\cite{SC_TIFS}.
Note that with only the statistical CSIT of both main and
eavesdropper's channel, the capacity and optimal signaling of the
fast fading wiretap channel were found in \cite{SC_TIFS}. Note also
that the wiretap channel is a special case of the GBC-CM which can
be easily derived by letting $W_2$ in the GBC-CM as null. The
phenomenon of channel degradation was also observed in
\cite{Liu_SISO} for the less noisy GBC-CM case with full CSIT, which
is a general case of the SISO GBC-CM. However, it is much more
involved to analyze the channel degradation phenomenon for partial
CSIT, i.e., the same marginal property needs to be re-examined. Thus
before illustrating our first main result, we need to introduce the
following same marginal lemma first, which is critical to the
result. Since Lemma \ref{Def_same_marginal} can be easily extended
from \cite[Lemma 4]{Liu_MISO} by treating the CSI as the channel
output \cite{caire_channel_with_SI}, we do not expose the proof
here.}
\begin{lemma}\label{Def_same_marginal}
Let $\mathcal{P}$ denote the set of channels $p(\tilde{y}_1,\,\tilde{y}_2,\,\tilde{\bm{h}}_1,\,\tilde{\bm{h}}_2|\bm{x})$ whose marginal distributions satisfy
\begin{align}
&p_{\tilde{Y}_1,\tilde{\bm{H}}_1|\bm{X}}(\tilde{y}_1,\tilde{\bm{h}}_1|\bm{x})=p_{Y_1,\bm{H}_1|\bm{X}}(y_1,\bm{h}_1|\bm{x}),\\
&p_{\tilde{Y}_2,\tilde{\bm{H}}_2|\bm{X}}(\tilde{y}_2,\tilde{\bm{h}}_2|\bm{x})=p_{Y_2,\bm{H}_2|\bm{X}}(y_2,\bm{h}_2|\bm{x}),
\end{align}
for all $y_1,\,y_2,\, \mbox{and } \bm{x}$. The secrecy capacity region is the same for all channels $p(\tilde{y}_1,\,\tilde{y}_2,\,\tilde{\bm{h}}_1,\,\tilde{\bm{h}}_2|\bm{x})\in\mathcal{P}$.
\end{lemma}

Note that $p(\tilde{y}_1,\,\tilde{y}_2,\,\tilde{\bm{h}}_1,\,\tilde{\bm{h}}_2|\bm{x})$ can be factorized as the statement below \eqref{EQ_R2_DMC}. Then our first result is as following.
\begin{lemma}\label{Lemma_necessary_conditon_FMGBCCM}
A necessary condition for both users in the fast Rayleigh FMGBC-CM
having positive rates is that the covariance matrices of $\bm{H}_1$
and $\bm{H}_2$ are not scaled of each other.
\end{lemma}
\begin{proof} We prove it by contradiction. Assume that the two channels are distributed as
$\bm{H}_1\sim CN(0,\sigma_1^2\bm\Sigma)$ and $\bm{H}_2\sim
CN(0,\sigma_2^2\bm\Sigma)$, i.e., the covariance matrices are scaled
of each other. Due to the same marginal property in Lemma
\ref{Def_same_marginal}, we can replace $\bm{H}_1$ in \eqref{EQ_Eve}
by $\tilde{\bm{H}}_1=(\sigma_1/\sigma_2)\bm{H}_2$ without affecting
the capacity. After that we have a  new pair of channels with the
same capacity as \eqref{EQ_Eve}
\begin{align}
Y_1'=(\sigma_1/\sigma_2)\bm{H}_2^H\bm{X}+N_1,\,\,\,Y_2=\bm{H}_2^H\bm{X}+N_2,\notag
\end{align}
which can be further represented as
\begin{align}
Y_1''=\bm{H}_2^H\bm{X}+(\sigma_2/\sigma_1)N_1,\,\,\,
Y_2=\bm{H}_2^H\bm{X}+N_2.\notag
\end{align}
Thus as long as $\sigma_1>\sigma_2$, we can have the Markov chain
$\bm{X}\rightarrow Y_1''\rightarrow Y_2$. On the other hand, by
extending the outer bound of \cite[Theorem 3]{Liu_SISO}, we know
that \textit{less noisy} \cite[Ch. 5]{Kim_lecture} makes
 one of the two-user FMGBC-CM have zero rate. Since $Y_2$ is degraded of $Y_1''$, and the degradedness property is stricter than the less noisy,
we can conclude that $R_2=0$. Similarly, when $\sigma_2>\sigma_1$,
$R_1=0$.
\end{proof}
\vspace{1cm}
An intuitive explanation is that, if a message can be successfully decoded by the inferior user, then the superior user is also ensured of decoding it. Thus the secrecy rate of the degraded user is zero. Therefore to avoid the investigation of such cases, in the following we assume $\bm{H}_1\sim CN(\bm\mu_1,\bK_{\bm{H}_1})$ and $\bm{H}_2\sim CN(\bm\mu_2,\bK_{\bm{H}_2})$, where $\bK_{\bm{H}_1}$ and $\bK_{\bm{H}_2}$ may not be scaled of each other. Two special cases with single input single output (SISO) antenna GBC-CM are also summarized as follows.

\begin{coro}
All SISO fast Rayleigh fading GBC-CMs with only statistical CSIT
degrade as wiretap channels.
\end{coro}

\begin{coro}
All SISO fast Rician fading GBC-CMs with only statistical CSIT
degrade as wiretap channels if the channels have the same
$K$-factor.
\end{coro}

\textit{Remark:} To prove the degradedness of the FMGBC-CM, another
way is to directly verifying the \textit{less noisy}
\cite{Kim_lecture} property, i.e., $I(\bm{V}_k,Y_1)>
I(\bm{V}_k,Y_2)$ or $I(\bm{V}_k,Y_1)< I(\bm{V}_k,Y_2)$, where
$k=1,\,2$, from extending the upper bound derivation in
\cite[Theorem 3, Example 1]{Liu_SISO} to multiple antenna case.
However, the derivation is intractable.

\section{The achievable secrecy rate region of FMGBC-CM}\label{Sec_rate_region_MGBC_CM}

Due to the fact that there is only statistical CSIT, we can not use
the original minimum mean square error (MMSE) inflation factor as
Costa \cite{CostaDPC}, where the exact channel state information is
required. Thus we need to re-derive the achievable rate region of
the FMGBC-CM instead of directly using Liu's result in \cite[Lemma
3]{Liu_MISO}. To derive the new achievable rate region, we resort to
the linear assignment Gel'fand-Pinsker coding (LA-GPC)
\cite{Gelfand_Pinsker}\cite{pslin_CR} with Gaussian codebooks, which
is the generalized case of DPC, to deal with the fading channels.
First, separate the channel input $\bm{X}$ into two random vectors
$\bm{U}_1$ and $\bm{U}_2$ so that $\bm{X}=\bm{U}_1+\bm{U}_2$. Then
$\bm{U}_1$ and $\bm{U}_2$ are chosen as
\begin{align}
&\bm{U}_1\sim CN(\mathbf{0},\mathbf{K}_{\bm{U}_1}),\,\bm{U}_2\sim CN(\mathbf{0},\mathbf{K}_{\bm{U}_2}),
\end{align}
where $\bm{U}_2$ is independent of $\bm{U}_1$, $\mathbf{K}_{\bm{U}_1}\succeq 0$ and $\mathbf{K}_{\bm{U}_2}\succeq 0$ are the covariance matrices of $\bm{U}_1$ and $\bm{U}_2$, respectively. After that, we do the decomposition $\mathbf{K}_{\bm{U}_1}=\mathbf{T}_1 \mathbf{T}_1^H$, and define $\bm{U}_1' \sim CN(\mathbf{0},\mathbf{I}_N)$ so that $\bm{U}_1=\mathbf{T}_1\bm{U}_1'$, where $\bT_1\in\mathds{C}^{n_T\times N}$ and $N$ is the rank of $\mathbf{K}_{\bm{U}_1}$.
The auxiliary random variables of LA-GPC are then defined as
\begin{align}
&\bm{V}_1=\bm{U}_1'+\ba\bm{H}_1^H\bm{U}_2, \label{dpc}\\
&\bm{V}_2=\bm{U}_2,
\end{align}
where $\mathbf{a}$ is the inflation factor in LA-GPC. The reason of
choosing \eqref{dpc} is that, if we do LA-GPC for $\bm{U}_1$
directly, i.e., $\bm{V}_1=\bm{U}_1+\ba\bm{H}_1^H\bm{U}_2$, but not
\eqref{dpc}. After substituting it into the RHS of \eqref{EQ_R2_DMC}
and \eqref{EQ_R2_DMC2}, we can find that when calculating
$I(\bm{V}_1;\bm{V}_2)$, the term the term
$\log|\mathbf{K}_{\bm{U}_1}|$ requires $\mathbf{K}_{\bm{U}_1}\succ
0$ but not $\mathbf{K}_{\bm{U}_1}\succeq 0$. However, the expression
of \eqref{dpc} would avoid this constraint.
Note that in the rest of this paper, for convenience the of derivation, we combine $\ba\bm{H}_1^H$ as $\mathbf{b}$. To present the rate regions compactly, recall that the permutation $\pi$ specifies the encoding order, i.e., the message of user $\pi_2$ is encoded first while the message of user $\pi_1$ is encoded second. Then we can have the following secrecy rate region.\\

\begin{lemma} \label{Lemma_rate_region}
 Let $\mathcal{R}(\mathbf{K}_{\bm{U}_{\pi_1}},\mathbf{K}_{\bm{U}_{\pi_2}})$ denote the union of all rate pairs $(R_{\pi_1},R_{\pi_2})$ satisfying
\begin{align}
R_{\pi_1} &\leq \left( E_{\bm{H}_{\pi_1}}[\log(1+{\bm{H}_{\pi_1}}^H(\mathbf{K}_{\bm{U}_{\pi_1}}+\mathbf{K}_{\bm{U}_{\pi_2}}){\bm{H}_{\pi_1}})]-\bigtriangleup\right)^+,\label{EQ_R1_GBC}\\
R_{\pi_2} &\leq \left(  E_{\bm{H}_{\pi_2}}[\log(1+{\bm{H}_{\pi_2}}^H(\mathbf{K}_{\bm{U}_{\pi_1}}+\mathbf{K}_{\bm{U}_{\pi_2}}){\bm{H}_{\pi_2}})]-\bigtriangleup\right)^+\label{EQ_R2_GBC},
\end{align}
where
\begin{align}
\bigtriangleup&\triangleq
E_{\bm{H}_{\pi_2}}[\log(1+{\bm{H}_{\pi_2}}^H\mathbf{K}_{\bm{U}_{\pi_1}}{\bm{H}_{\pi_2}})]+E_{\bm{H}_{\pi_1}}\!\!\left[\log\left|\!\!
\begin{array}{cc}
                                                                                                   { \mathbf{I}+\mathbf{bK}_{\mathbf{U}_{\pi_2}}\mathbf{b}^H }&{ (\mathbf{T}_{\pi_1}^H+\mathbf{bK}_{\mathbf{U}_{\pi_2}}){\bm{H}_{\pi_1}} }  \\
                                                                                                   { {\bm{H}_{\pi_1}}^H(\mathbf{T}_{\pi_1}+\mathbf{K}_{\bm{U}_{\pi_2}}\mathbf{b}^H) }&{ 1+{\bm{H}_{\pi_1}}^H(\mathbf{K}_{\bm{U}_{\pi_1}}+\mathbf{K}_{\bm{U}_{\pi_2}}){\bm{H}_{\pi_1}} }
                                                                                                  \end{array}\label{EQ_Delta}
\!\!   \right|   \right]\\
&\triangleq
E_{\bm{H}_{\pi_2}}[\log(1+{\bm{H}_{\pi_2}}^H\mathbf{K}_{\bm{U}_{\pi_1}}{\bm{H}_{\pi_2}})]+E_{\bm{H}_{\pi_1}}\!\!\left[\log\left|\mathbf{M}\right|\right].\label{EQ_def_M}
\end{align}
Then any rate pair
\[
(R_1,R_2)\in \mbox{co}\left\{\underset{\mbox{tr}(\mathbf{K}_{\bm{U}_{\pi_1}}+\mathbf{K}_{\bm{U}_{\pi_2}})\leq P_T}\bigcup\mathcal{R}(\mathbf{K}_{\bm{U}_{\pi_1}},\mathbf{K}_{\bm{U}_{\pi_2}})\right\}
\]
is achievable for the FMGBC-CM.\\
\end{lemma}
The proof is provided in Appendix \ref{Sec_App_Lemma3}. In the
following, we provide an achievable scheme to approximately achieve
the above two bounds in \eqref{EQ_R1_GBC} and \eqref{EQ_R2_GBC}.
Assume $\bm{H}_{\pi_k}\sim CN(\bm \mu_{\pi_k},\,\bm K_{\bm
H_{\pi_k}})$, $k=1,2$.

\begin{lemma}\label{lemma_optimal_input_cov_mat}
With the selection $\mathbf{K}_{\bm{U}_{\pi_1}}^*=\alpha P_T\mathbf{e}_{\pi_1}^*(\mathbf{e}_{\pi_1}^*)^H$ and $\mathbf{K}_{\bm{U}_{\pi_2}}^*=(1-\alpha) P_T\mathbf{e}_{\pi_2}^*(\mathbf{e}_{\pi_2}^*)^H$, where $||\mathbf{e}_{\pi_1}^*||^2=||\mathbf{e}_{\pi_2}^*||^2=1$, and
\begin{align}
\mathbf{e}_{\pi_1}^*&=\arg\max_{\mathbf{e}_{\pi_1}}\frac{\mathbf{e}_{\pi_1}^H\left(\mathbf{I}+\alpha P_T\left(\bK_{{\bm{H}}_{\pi_1}}+\bm\mu_{\pi_1}\bm\mu_{\pi_1}^H\right)\right)\mathbf{e}_{\pi_1}}{\mathbf{e}_{\pi_1}^H\left(\mathbf{I}+\alpha P_T\left(\bK_{{\bm{H}}_{\pi_2}}+\bm\mu_{\pi_2}\bm\mu_{\pi_2}^H\right)\right)\mathbf{e}_{\pi_1}},\label{EQ_e1}
\\
\mathbf{e}_{\pi_2}^*&=\arg\max_{\mathbf{e}_{\pi_2}}\frac{\mathbf{e}_{\pi_2}^H\left(\mathbf{I}+\frac{(1-\alpha)P_T\left(\bK_{{\bm{H}}_{\pi_2}}+\bm\mu_{\pi_2}\bm\mu_{\pi_2}^H\right)}{1+\alpha P_T(\mathbf{e}_{\pi_1}^*)^H\left(\bK_{{\bm{H}}_{\pi_2}}+\bm\mu_{\pi_2}\bm\mu_{\pi_2}^H\right)\mathbf{e}_{\pi_1}^*}\right)
\mathbf{e}_{\pi_2}}{\mathbf{e}_{\pi_2}^H\left(\mathbf{I}+\frac{(1-\alpha)P_T\left(\bK_{{\bm{H}}_{\pi_1}}+\bm\mu_{\pi_1}\bm\mu_{\pi_1}^H\right)}{1+\alpha P_T(\mathbf{e}_{\pi_1}^*)^H\left(\bK_{{\bm{H}}_{\pi_1}}+\bm\mu_{\pi_1}\bm\mu_{\pi_1}^H\right)\mathbf{e}_{\pi_1}^*}\right)\mathbf{e}_{\pi_2}},\label{EQ_e2}
\end{align}
where $\alpha$ is the ratio of power allocated to user $\pi_1$, we can get the non-trivial rate region for the FMGBC-CM as
\[
(R_1,R_2)\in \mbox{co}\left\{\underset{0\leq\alpha\leq 1}\bigcup\mathcal{R}(\mathbf{K}_{\bm{U}_{\pi_1}}^*,\mathbf{K}_{\bm{U}_{\pi_2}}^*)\right\}.\\
\]
\end{lemma}

\vspace{1cm} The proof is provided in Appendix \ref{Sec_App_Lemma4}.
Note that with \cite[Property 2 and 3]{Petropulu_MIMOME} it can be
proved that when the number of transmit antenna is 2 with
$\bK_{\bm{H}_{\pi_1}}-\bK_{\bm{H}_{\pi_2}}$ being non-positive
semi-definite, unit rank $\bK_{\bm{U}_{\pi_1}}$ and
$\bK_{\bm{U}_{\pi_2}}$ are optimal for the considered upper bounds
of $R_1$ and $R_2$ in the proof of Lemma \ref{lemma_optimal_input_cov_mat}.\\

 After deriving the covariance matrices, we then need to solve the inflation factor due
 to the fact that there is no full CSIT. Here we resort to the following fixed point iteration to solve $\bb$
\begin{align}\label{EQ_iterative_b}
& \mathbf{b}=-(E_{\bm{H}_{\pi_1}}[\mathbf{A}_1^H])^{-1}E_{\bm{H}_{\pi_1}}[\mathbf{A}_2^H\bm{H}_1^H]\triangleq f(\bb),
 \left[ \begin{array}{c}
                                                                                                              \mathbf{A}_1 \\
                                                                                                              \mathbf{A}_2 \\
                                                                                                            \end{array}
                                                                                                          \right]\triangleq\mathbf{M}^{-1}\left[
                                                                                                                    \begin{array}{c}
                                                                                                                     \mathbf{I} \\
                                                                                                                      \mathbf{0} \\
                                                                                                                    \end{array}\right],
\end{align}
where $\mathbf{ M}$ is defined as the block matrix inside the
determinant of the second term in \eqref{EQ_def_M}. Note that
\eqref{EQ_iterative_b} is derived by $\partial R_1/\partial \bb=0$.
Note also that the iteration stops when the maximum of the absolute
values of the relative errors of $R_1$ and $R_2$ in the successive
iterations is less than a predefined value. The iteration steps are
summarized in Table \ref{TA iterative steps}.

\begin{table} [ht]
\begin{center}
\caption{The iterative steps for solving $\bb$.}
\begin{tabular}{l l} \label{TA iterative steps}
Step 1 & Set $i=0$ and initialize $\mathbf{b}^{(i)}=\mathbf{0}$. Also initialize $\be_{\pi_1}$ and $\be_{\pi_2}$ as \notag\\
&\eqref{EQ_e1} and \eqref{EQ_e2}, respectively.\\
Step 2 & Evaluate $\mathbf{b}^{(i+1)}=f(\mathbf{b}^{(i)})$ as \eqref{EQ_iterative_b}.\\
Step 3 & Let $i=i+1$ and repeat Step 2 until\notag\\
& $\max\left\{\Big|R_1^{(i)}-R_1^{(i-1)}\Big|,\,\Big|R_2^{(i)}-R_2^{(i-1)}\Big|\right\}<\delta_1$.
\end{tabular}
\end{center}
\end{table}
\section{The Iterative Algorithm}
From Lemma \ref{Lemma_rate_region} we can observe that the achievable rate region requires the joint
optimization of the inflation factor and input covariance matrices, i.e., $\mathbf{b}$,
 $\mathbf{K}_{\bm{U}_{\pi_1}}$, and $\mathbf{K}_{\bm{U}_{\pi_2}}$, or equivalently, $\mathbf{b}$,
  $\mathbf{T}_{{\pi_1}}$, and $\mathbf{T}_{{\pi_2}}$. However, the largest rate region is in
  general difficult to find unless we prove it reaches the upper bound of the capacity region.
  Instead, here we develop an algorithm to provide a sub-optimal solution for this optimization
   problem. To be more specific, given the ratio of total power allocated to user $\pi_1$, i.e.,
    $\alpha=\tr{(\mathbf{T}_{{\pi_1}}\mathbf{T}_{{\pi_1}}^H)}/P_T$ and the achievable rate region
    as the union of \eqref{EQ_R1_GBC} and \eqref{EQ_R2_GBC} among all $\alpha \in [0,1]$, we
     instead resort to solving the following two optimization problems $\bm P_1$
     and $\bm P_2$, iteratively
\begin{align}
&\bm P_1:\,\,\max_{\mathbf{b},\mathbf{T}_{{\pi_1}}}{R_{\pi_1}^{UB}},\,\,
\mbox{s.t.}~
\tr{(\mathbf{T}_{{\pi_1}}\mathbf{T}_{{\pi_1}}^H)}\leq \alpha P_T,\label{EQ_opt1}\\
&\bm P_2:\,\,\max_{\mathbf{T}_{{\pi_2}}}{R_{\pi_2}^{UB}},\,\,
\mbox{s.t.}~
\tr{(\mathbf{T}_{{\pi_2}}\mathbf{T}_{{\pi_2}}^H)}\leq (1-\alpha)P_T,\label{EQ_opt2}
\end{align}
where $R_{\pi_1}^{UB}$ and $R_{\pi_2}^{UB}$ are the RHS of the rate formulae in \eqref{EQ_R1_GBC} and \eqref{EQ_R2_GBC}, respectively, and in problems $\bm P_1$ and $\bm P_2$, $\mathbf{T}_{{\pi_2}}$ and $(\mathbf{T}_{{\pi_1}},\,\mathbf{b})$ are given, respectively.

In the following, we show the functions in terms of $\mathbf{b}$,
$\mathbf{K}_{\bm{U}_{\pi_1}}$, and $\mathbf{K}_{\bm{U}_{\pi_2}}$,
which will be used to develop the proposed fixed point iterative
algorithm in solving $\bm P_1$ and $\bm P_2$ later, derived from the Karush-Kuhn-Tucker (K.K.T.) conditions. \\

\begin{lemma}\label{lemma_KKT}
The K.K.T. necessary conditions of the optimization problems $\bm
P_1$ and $\bm P_2$ can be formed as
$f_1(\mathbf{b},\mathbf{T}_{\pi_1},\mathbf{T}_{\pi_2})=\mathbf{b},\,\,
g_1(\mathbf{b},\mathbf{T}_{\pi_1},\mathbf{T}_{\pi_2})=\lambda_1\mathbf{T}_{\pi_1},\,\,g_2(\mathbf{b},\mathbf{T}_{\pi_1},\mathbf{T}_{\pi_2})=\lambda_2\mathbf{T}_{\pi_2},$
where $\lambda_1$ and $\lambda_2$ are Lagrange multipliers of the
Lagrangians corresponding to \eqref{EQ_opt1} and \eqref{EQ_opt2},
respectively, and $f_1$, $g_1$, and $g_2$ are defined as
\begin{align}
f_1(\mathbf{b},\mathbf{T}_{\pi_1},\mathbf{T}_{\pi_2})\overset{\Delta}=&-(E_{\bm{H}_{\pi_1}}[\mathbf{A}_1^H])^{-1}E_{\bm{H}_{\pi_1}}[\mathbf{A}_2^H\bm{H}_{\pi_1}^H], \notag\\
g_1(\mathbf{b},\mathbf{T}_{\pi_1},\mathbf{T}_{\pi_2})\overset{\Delta}=&E_{\bm{H}_{\pi_1}}\!\!\!\!\left[\bm{H}_{\pi_1}\!\!\left(\!\!1\!\!+\!\!\bm{H}_{\pi_1}^H(\mathbf{K}_{\bm{U}_{\pi_1}}\!\!\!\!+\!\!\mathbf{K}_{\bm{U}_{\pi_2}}\!\!)\bm{H}_{\pi_1}\right)^{-1}\!\!\bm{H}_{\pi_1}^H\!\!\mathbf{T}_{\pi_1}\!\!\right] -E_{\bm{H}_{\pi_2}}\left[\bm{H}_{\pi_2}\left(1+\bm{H}_{\pi_2}^H\mathbf{K}_{\bm{U}_{\pi_1}}\bm{H}_{\pi_2}\right)^{-1}\bm{H}_{\pi_2}^H\mathbf{T}_{\pi_1}\right] \notag\\
&-E_{\bm{H}_{\pi_1}}\left[[\mathbf{0}~~ \bm{H}_{\pi_1}]\mathbf{M}^{-H}\left[
                            \begin{array}{c}
                              \mathbf{I} \\
                              \bm{H}_{\pi_1}^H\mathbf{T}_{\pi_1} \\
                            \end{array}
                          \right]
\right], \notag\\
g_2(\mathbf{b},\mathbf{T}_{\pi_1},\mathbf{T}_{\pi_2})\overset{\Delta}=&E_{\bm{H}_{\pi_2}}\!\!\!\!\left[ \mathbf{h}_{\pi_2}\!\!\left(\!\!1\!\!+\!\!\bm{H}_{\pi_2}^H\!\!(\mathbf{K}_{\bm{U}_{\pi_1}}\!\!+\!\!\mathbf{K}_{\bm{U}_{\pi_2}})\bm{H}_{\pi_2}\right)^{-1}\!\!\bm{H}_{\pi_2}^H\!\!\mathbf{T}_{\pi_2} \!\!\right] -E_{\bm{H}_{\pi_1}}\left[\left[\begin{array}{cc}\mathbf{b}^H &\bm{H}_{\pi_1} \\ \end{array}\right] \mathbf{M}^{-H}\left[
                                                     \begin{array}{c}
                                                       \mathbf{b} \\
                                                       \bm{H}_{\pi_1}^H \\
                                                     \end{array}
                                                   \right]\mathbf{T}_{\pi_2}
 \right], \notag
\end{align}
where $\bm A_1$, $\bm A_2$, and $\mathbf{M}$ are defined in \eqref{EQ_iterative_b}.
\end{lemma}
\vspace{1cm} The proof is given in Sec \ref{Sec_App_b}. Obviously,
the analytical solutions for the equations $f_1$, $g_1$, and $g_2$
are intractable. We then rather propose an iterative algorithm to
solve this joint optimization problem, which is summarized in the
following.
\begin{center}
\line(2,0){510}\\\vspace{-0.8cm}\line(2,0){510}
\end{center}
\begin{description}
\item[Step 1. ]  \hspace{0.2cm}Initialize $\mathbf{b}^{(0)}$, $\mathbf{T}_{\pi_1}^{(0)}$ and $\mathbf{T}_{\pi_2}^{(0)}$ randomly, such that tr$\left(\mathbf{T}_{\pi_1}^{(0)}\left(\mathbf{T}_{\pi_1}^{(0)}\right)^H\right)$+tr$\left(\mathbf{T}_{\pi_2}^{(0)}\left(\mathbf{T}_{\pi_2}^{(0)}\right)^H\right)\leq P_T$.
\item[Step 2. ]\hspace{0.2cm}For the ($i$+1)-th iteration,
\hspace{0.5cm}\item[2.A.] Initialize $\mathbf{b}_{in}^{(0)}=\mathbf{b}^{(i)}$ and run iterations
                  $\mathbf{b}_{in}^{(j+1)}=f_1\left(\mathbf{b}_{in}^{(j)},\mathbf{T}_{\pi_1}^{(i)},\mathbf{T}_{\pi_2}^{(i)}\right)$ till $|\mathbf{b}_{in}^{(j+1)}-\mathbf{b}_{in}^{(j)}|\leq\epsilon_1$. Then set $\mathbf{b}^{(i+1)}$ as the fixed point solution of $\mathbf{b}_{in}$.
\item[2.B.] Initialize $\mathbf{T}_{\pi_{1,\,in}}^{(0)}=\mathbf{T}_{\pi_1}^{(i)}$. Run iterations till $|\mathbf{T}_{\pi_{1,\,in}}^{(j+1)}-\mathbf{T}_{\pi_{1,\,in}}^{(j)}|\leq\epsilon_2$:
                  $\mathbf{T}_{\pi_{1,\,in}}^{(j+1)}=\frac{1}{\lambda_1^{(j+1)}}g_1\left(\mathbf{b}^{(i+1)},\mathbf{T}_{\pi_{1,\,in}}^{(j)},\mathbf{T}_{\pi_2}^{(i)}\right)$ and set $\mathbf{T}_{\pi_{1,\,in}}^{(j+1)}=0.5\mathbf{T}_{\pi_{1,\,in}}^{(j+1)}+0.5\mathbf{T}_{\pi_{1,\,in}}^{(j)}$.
                   The Lagrange multiplier is calculated by \\ $\lambda_1^{(j+1)}=\sqrt{\tr{\left(g_1\left(\mathbf{b}^{(i+1)},\mathbf{T}_{\pi_{1,\,in}}^{(j)},\mathbf{T}_{\pi_2}^{(i)}\right)\left(g_1\left(\mathbf{b}^{(i+1)},\mathbf{T}_{\pi_{1,\,in}}^{(j)},\mathbf{T}_{\pi_2}^{(i)}\right)\right)^H\right)}\Bigg/\tr{\left(\mathbf{T}_{\pi_{1,\,in}}^{(j)}\left(\mathbf{T}_{\pi_{1,\,in}}^{(j)}\right)^H\right)}}$. Then set $\mathbf{T}_{\pi_1}^{(i+1)}$as the fixed point solution of $\mathbf{T}_{\pi_{1,\,in}}$.
\item[2.C.] Initialize $\mathbf{T}_{\pi_{2,\,in}}^{(0)}=\mathbf{T}_{\pi_2}^{(i)}$. Run iterations till $|\mathbf{T}_{\pi_{2,\,in}}^{(j+1)}-\mathbf{T}_{\pi_{2,\,in}}^{(j)}|\leq\epsilon_3$:
                  $\mathbf{T}_{\pi_{2,\,in}}^{(j+1)}=\frac{1}{\lambda_2^{(j+1)}}g_2\left(\mathbf{b}^{(i+1)},\mathbf{T}_{\pi_1}^{(i+1)},\mathbf{T}_{\pi_{2,\,in}}^{(j)}\right)$ and set $\mathbf{T}_{\pi_{2,\,in}}^{(j+1)}=0.5\mathbf{T}_{\pi_{2,\,in}}^{(j+1)}+0.5\mathbf{T}_{\pi_{2,\,in}}^{(j)}$. The Lagrange multiplier is calculated by \\ $\lambda_2^{(j+1)}=\sqrt{\tr{\left(g_2\left(\mathbf{b}^{(i+1)},\mathbf{T}_{\pi_1}^{(i+1)},\mathbf{T}_{\pi_{2,\,in}}^{(j)}\right)\left(g_2\left(\mathbf{b}^{(i+1)},\mathbf{T}_{\pi_1}^{(i+1)},\mathbf{T}_{\pi_{2,\,in}}^{(j)}\right)\right)^H\right)}\Bigg/\tr{\left(\mathbf{T}_{\pi_{2,\,in}}^{(j)}\left(\mathbf{T}_{\pi_{2,\,in}}^{(j)}\right)^H\right)}}$. Then set $\mathbf{T}_{\pi_2}^{(i+1)}$ as the fixed point  solution of $\mathbf{T}_{\pi_{2,\,in}}$.

\item[Step 3.  ]  $\,\,\,\,$Repeat Step 2. until
\[
\max\left\{\Big|R_{\pi_1}^{UB,(i+1)}-R_{\pi_1}^{UB,(i)}\Big|,\,\Big|R_{\pi_2}^{UB,(i+1)}-R_{\pi_2}^{UB,(i)}\Big|\right\}\leq \delta.
\]
\end{description}
\begin{center}
\line(2,0){510}\\\vspace{-0.8cm}\line(2,0){510}
\end{center}
The above steps are also summarized in Fig.\ref{Fig_flow_chart_alg}. Note that each step among Step 2.A to Step 2.C
has a local iterative process with intermediate variables $\mathbf{b}_{in}$, $\mathbf{T}_{\pi_{1,\,in}}$, and
$\mathbf{T}_{\pi_{2,\,in}}$, respectively. Note also that in Step 2.B and 2.C, we update $\mathbf{T}_{\pi_{1,\,in}}$
and $\mathbf{T}_{\pi_{2,\,in}}$ with moving average. The reason is that without the moving average the original algorithm
 is sensitive to initial values and may be easily trapped in a bad solution. However, this condition can be improved with
  such modification.

\section{ Low SNR Analysis}\label{Sec_low_SNR}
In this section we study the achievable secrecy rate region in the low-SNR regime. The motivations of this study are
as following. First, note that operation at low SNRs is beneficial from a security perspective since it is generally
 difficult for an eavesdropper to detect the signal. Second, it is well-known that for fading Gaussian channels
 subject to average input power constraints, energy efficiency improves as one operates at low SNR level, and the
 minimum bit energy is achieved as SNR vanished \cite{Gursoy_security_low_SNR}. Therefore, with the aid low SNR analysis we can determine the
 best energy efficiency of our model, which can be measured by the minimum energy required to send one information bit reliably. Finally, due to the rate region in Lemma
 \ref{Lemma_rate_region} is too complicated to analyze and by asymptotic analysis the rate formulae of most
 channels can be highly simplified, we resort to the low SNR regime to get some insights { of the optimality of choosing the input covariance matrices as unit rank
 }. In the following,
 we first characterize the secrecy rate region of FMGBC-CM under low SNR regime. Then we analyze the minimum
 bit energy. For the convenience of discussion, in the following we assume that the channel is fast Rayleigh
 faded, which can be easily extended to the Rician channels.\\
\begin{lemma} \label{lemma_low_snr}
In the low SNR regime, the optimal input covariance matrices $\mathbf{K}_{\bm{U}_1}$ and $\mathbf{K}_{\bm{U}_2}$ are both unit rank, with the directions aligned to the eigenvector corresponding to the maximum eigenvalues of $\mathbf{K}_{\bm{H}_1}-\mathbf{K}_{\bm{H}_2}$ and $\mathbf{K}_{\bm{H}_2}-\mathbf{K}_{\bm{H}_1}$, respectively. And the asymptote of the secrecy rate region is
\begin{align}
R_1 &\leq \left(\frac{\alpha P}{\ln 2}\lambda_{max}(\mathbf{K}_{\bm{H}_1}-\mathbf{K}_{\bm{H}_2})\right)^+, \label{R1_low_SNR}\\
R_2 &\leq \left(\frac{(1-\alpha) P}{\ln 2}\lambda_{max}(\mathbf{K}_{\bm{H}_2}-\mathbf{K}_{\bm{H}_1})\right)^+.\label{R2_low_SNR}
\end{align}
\end{lemma}

Before proving Lemma \ref{lemma_low_snr}, we need the following lemma which provides some useful mathematical properties for the proof.

\begin{lemma} \label{math_low}
Given a symmetric matrix $\mathbf{A}_{n\times n}$ and assume
$\mathbf{M}_{n\times n}\succeq 0$. The optimal $\mathbf{M}$ which
maximizes $\tr(\mathbf{AM})$ under the constraint
$\tr(\mathbf{M})\leq 1$ is unit rank. Furthermore, the eigenvector
of the optimal $\mathbf{M}$ is the one corresponding to the maximum
eigenvalue of ${\mathbf{A}}$.
\end{lemma}

The proofs of Lemma \ref{lemma_low_snr} and Lemma \ref{math_low} are given in Sec. \ref{sec_app_lemma6} and Sec. \ref{sec_app_lemma7}, respectively. Note that the unit rank result in Lemma \ref{lemma_low_snr} is consistent to that of MGBC-CM with perfect CSIT, also our selection of $\bK_{\bm{U}_{\pi_1}}^*$ and $\bK_{\bm{U}_{\pi_2}}^*$ in Lemma \ref{lemma_optimal_input_cov_mat}. Besides, the rate region described in (\ref{R1_low_SNR}) and (\ref{R2_low_SNR}) also implies that
\begin{coro}\label{coro_low_snr_positive_rates_condition}
In the low SNR regime, both users can have positive rates simultaneously if and only if $\mathbf{K}_{\bm{H}_1}-\mathbf{K}_{\bm{H}_2}$ is indefinite.
\end{coro}

Note that in the low SNR regime we can have a stronger result as
Corollary \ref{coro_low_snr_positive_rates_condition} than Lemma
\ref{Lemma_necessary_conditon_FMGBCCM} in the normal SNR regime, in
the sense that Corollary \ref{coro_low_snr_positive_rates_condition}
is a necessary and sufficient condition. In the last part of this
section, we would like to measure the best energy efficiency of our
model, i.e. the minimum bit energy. As mentioned previously, this
happens when SNR approaches to zero (or equivalently, signal power
$P$ approximates to zero with fixed noise power $N_0$)
\cite{Verdu_low_SNR}, which is defined as
\begin{align}
\left(\frac{E_b}{N_0}\right)_{min}\triangleq\lim_{P\rightarrow 0}{\frac{P}{R(P)}}=\frac{1}{\dot{R}(0)}. \label{efficiency}
\end{align}
The above equality comes from the fact that $\lim_{P\rightarrow 0}{R(P)}=\dot{R}(0)P+O(P^2)$, where $\dot{R}(0)$ denotes the first derivative of $R$ at $P=0$. Hence, after substituting the results of (\ref{R1_low_SNR}) and (\ref{R2_low_SNR}) into (\ref{efficiency}), we have
\begin{align}
\left(\frac{E_b^{s,\,SCSIT}}{N_0}\right)_{R_1,\,\,min}=\frac{1}{\dot{R_1}(0)}=\frac{\ln 2}{\lambda_{max}},\,\,\,\left(\frac{E_b^{s,\,SCSIT}}{N_0}\right)_{R_2,\,\,min}=\frac{1}{\dot{R_2}(0)}=\frac{\ln 2}{\lambda_{min}},
\end{align}
where $\lambda_{max}$ and $\lambda_{min}$ denote the maximum and minimum eigenvalues of $\mathbf{K}_{\bm{H}_1}-\mathbf{K}_{\bm{H}_2}$, respectively; the superscript \textit{s, SCSIT} denotes the secure communications with statistical CSIT (SCSIT).\\

To elucidate how the perfect secrecy constraint affects the best
energy efficiency, we first compare the minimum bit energy of
communications with and without secrecy constraints. Note that for
convenience, in the following we only consider the best energy
efficiency of User 1, and the results of User 2 can be directly
extended. Following the similar derivation of Lemma
\ref{lemma_low_snr}, the rates for communications in low SNR without
secrecy constraint can be easily obtained as $\frac{\ln
2}{\lambda_{max}(\mathbf{K}_{\bm{H}_1})}$. Thus we have
\begin{align}
\left(\frac{E_b^{s,\,SCSIT}}{N_{0}}\right)_{R_1,\,\,min}=\frac{\ln 2}{\lambda_{max}(\mathbf{K}_{\bm{H}_1}-\mathbf{K}_{\bm{H}_2})}\geq \frac{\ln 2}{\lambda_{max}(\mathbf{K}_{\bm{H}_1})}=\left(\frac{E_b^{SCSIT}}{N_{0}}\right)_{R_1,\,\,min},
\end{align}
where $\left(\frac{E_b^{SCSIT}}{N_0}\right)_{R_1,\,\,min}$ denote the optimal energy efficiency under SCSIT for communication \textit{without} secrecy constraints. And the above inequality comes from the fact that for matrices $\mathbf{A}\succcurlyeq 0$ and $\mathbf{B}\succcurlyeq 0$, $\lambda_{max}(\mathbf{A}-\mathbf{B})\leq\lambda_{max}(\mathbf{A})-\lambda_{min}(\mathbf{B})\leq\lambda_{max}(\mathbf{A})$. From the results, as expected, we can conclude that secrecy constraints increase the bit-energy requirements. And the increment is
\[
\left(\frac{E_b^{s,\,SCSIT}}{N_{0}}\right)_{R_1,\,\,min}-\left(\frac{E_b^{SCSIT}}{N_{0}}\right)_{R_1,\,\,min}=\ln 2\left(\frac{1}{\lambda_{max}(\mathbf{K}_{\bm{H}_1}-\mathbf{K}_{\bm{H}_2})}-\frac{1}{\lambda_{max}(\mathbf{K}_{\bm{H}_1})}\right)
\]

Finally, we demonstrate the impact of the knowledge of CSIT to best energy efficiency. The minimum bit energy for full CSIT case is $\ln2\left(E_{{\bm{H}_1},{\bm{H}_2}}\left[\left(\lambda_{max}(\bm{H}_1\bm{H}_1^H-\bm{H}_2\bm{H}_2^H)\right)^+\right]\right)^{-1}$, which can be directly extended from \cite{Gursoy_security_low_SNR}. Thus, we have the following relation
\begin{align}\label{EQ_energy_efficiency_FCSIT_SCSIT}
\left(\frac{E_b^{s,\,SCSIT}}{N_{0}}\right)_{R_1,\,\,min}&=\frac{\ln 2}{\lambda_{max}(\mathbf{K}_{\bm{H}_1}-\mathbf{K}_{\bm{H}_2})}\geq \frac{\ln 2}{E_{{\bm{H}_1},{\bm{H}_2}}\left[\left(\lambda_{max}(\bm{H}_1\bm{H}_1^H-\bm{H}_2\bm{H}_2^H)\right)^+\right]}=\left(\frac{E_b^{s,\,FCSIT}}{N_{0}}\right)_{R_1,\,\,min},
\end{align}
where $\left(\frac{E_b^{s,\,FCSIT}}{N_0}\right)_{R_1,\,\,min}$ denotes the optimal energy efficiency for secrecy communication with full CSIT (FCSIT). The inequality in \eqref{EQ_energy_efficiency_FCSIT_SCSIT} comes from the fact that the maximum eigenvalue of a symmetric matrix is convex \cite[Ch. 3]{Boyd_convex_optimization} and followed by applying the Jensen's inequality. The above inequality shows that the lack of CSIT indeed increases the energy expenditure.

Note that the discussions here not only indicate the relations of different communication systems in terms of the best energy efficiency but also provides an quantitative description for these systems. Therefore we can observe the impacts much more clearly.

\section{Numerical Results}\label{Sec_simulation}

In this section, we compare the rate regions solving from our
proposed iterative algorithm and the achievable transmission scheme
in Lemma \ref{lemma_optimal_input_cov_mat} under both Rayleigh and
Rician fading channels to that of  the MGBC-CM with full CSIT,
respectively. We set $n_T=2$ and the power constraint $P_T=10$,
respectively. The variances of all noises are set as 1. We also set
the stopping criteria of the iterative algorithm as
$\delta_1=\delta_2=\delta_3=\delta_4=\delta_5=10^{-3}$. In the
numerical simulation of fast Rayleigh fading case, we set the
covariance matrices of the two channels as
\begin{equation}\label{EQ_Kh_example}
\bK_{\bm{H}_1}=\left[\begin{array}{cc}
            0.2   &  0\\
             0  &  0.04
            \end{array}\right],\,\,\bK_{\bm{H}_2}=\left[\begin{array}{cc}
            0.1   &  0.08\\
             0.08  &  0.1
            \end{array}\right],
\end{equation}
which satisfy Lemma \ref{Lemma_necessary_conditon_FMGBCCM}. For the full CSIT case, we consider the rate region which is the convex closure of the following rate pair
\begin{align}
&R_{\pi_1}\leq E\left[\left(\log_{2}\frac{1+\bm{H}_{\pi_1}^H\bK_{\bm{U}_{\pi_1}}\bm{H}_{\pi_1}}{1+\bm{H}_{\pi_2}^H\bK_{\bm{U}_{\pi_1}}\bm{H}_{\pi_2}}\right)^+\right],\label{EQ_full_CSIT_R1}\\
&R_{\pi_2}\leq E\left[\left(\log_{2}\frac{[1+\bm{H}_{\pi_2}^H(\bK_{\bm{U}_{\pi_1}}+\bK_{\bm{U}_{\pi_2}})\bm{H}_{\pi_2}](1+\bm{H}_{\pi_1}^H\bK_{\bm{U}_{\pi_1}}\bm{H}_{\pi_1})}{[1+\bm{H}_{\pi_1}^H(\bK_{\bm{U}_{\pi_1}}+\bK_{\bm{U}_{\pi_2}})\bm{H}_{\pi_1}](1+\bm{H}_{\pi_2}^H\bK_{\bm{U}_{\pi_1}}\bm{H}_{\pi_2})}\right)^+\right],\label{EQ_full_CSIT_R2}
\end{align}
with the power constraint $\mbox{tr}(\bK_{\bm{U}_{\pi_1}}+\bK_{\bm{U}_{\pi_2}})\leq 10$, where the optimal $\bK_{\bm{U}_{\pi_1}}$ and $\bK_{\bm{U}_{\pi_2}}$ are described in \cite[(16)]{Liu_MISO} and the optimal $\bb$ is as
\begin{equation}\label{EQ_MMSE_b}
\bb=\bK_{\bm{U}_{\pi_1}}\bm{H}_{\pi_1}\bm{H}_{\pi_1}^H/(1+\bm{H}_{\pi_1}^H\bK_{\bm{U}_{\pi_1}}\bm{H}_{\pi_1}).
 \end{equation}
Note that \eqref{EQ_full_CSIT_R1} and \eqref{EQ_full_CSIT_R2} are the straightforward extension from \cite{Liu_MISO} to the fast fading
channels with full CSIT. From Fig. \ref{Fig_Rayleigh} we can easily see that both the proposed iterative algorithm and the transmission
 scheme in Lemma \ref{lemma_optimal_input_cov_mat} for the fast FMGBC-CM with partial CSIT apparently outperform the time sharing scheme.
  Time sharing means that the transmitter sends the two messages with different powers during a fraction of time where these powers
   satisfy the average power constraint. And in each fraction of time, the fast FMGBC-CM reduces to a fast fading Gaussian MISO wiretap
    channel. From this figure we can also find that the proposed transmission scheme proposed in Lemma \ref{lemma_optimal_input_cov_mat}
     performs better than the iterative algorithm when the ratio of the two users' secrecy rates is large enough. This is because the iterative algorithm may result in
     local optimal, which may be worse than the proposed transmission scheme. On the other hand,
      by comparing the regions of full and partial CSIT cases, we can easily find the impact of the knowledge of CSIT to the rate performance.
      That is,
without full knowledge of CSI, the transmitter is not able to design
signal directions and do power allocation efficiently. More
specifically, by comparing \eqref{EQ_R1_GBC}\eqref{EQ_R2_GBC} and
\eqref{EQ_full_CSIT_R1}\eqref{EQ_full_CSIT_R2}, respectively, we can
easily see that the operation $(.)^+$ in the former case is outside
the expectation but the later is inside. This is because in the
former case power and direction of channel input signals are fixed
and used for all channel realizations. Thus the rate loss is
inevitable. Similar phenomenon also emerges in
\cite{Khisti_MISOME}\cite{sclin_AN_security}.

For the Rician fading case, in addition to \eqref{EQ_Kh_example}, we let the mean vectors of $\bm{H}_1$ and $\bm{H}_2$ as
\[
\bm\mu_1=[   0.7,\,\,   0.1]^T,\,\,\,\bm\mu_2=[     0.1 ,\,\,    0.6]^T,
\]
respectively. From Fig. \ref{Fig_Rician},  we can easily see that
the CSIT plays an important role in improving the rate region in
this case. And time sharing is still the worst. Note that in this
case, the proposed transmission scheme in Lemma
\ref{lemma_optimal_input_cov_mat} also overwhelms the iterative
algorithm much more apparently than that in Rayleigh fading case.
With the aid of line of sight, the performances of all schemes under
Rician fading are much better than those under Rayleigh fading. To
illustrate the performance of $\mathbf{b}$ solved from Table \ref{TA
iterative steps}, in Fig. \ref{Fig_Rician_b} we compare the case
where $\bb$ is derived from substituting $\bm{H}_1=\bm\mu_1$ into
$\bb=\bT_{\pi_1}^H\bm{H}_{\pi_1}\bm{H}_{\pi_1}^H/(1+\bm{H}_{\pi_1}^H\bK_{\bm{U}_{\pi_1}}\bm{H}_{\pi_1}),$
which is worse than with $\bb$ solved from Table I. On the other
hand, due to the performance gap is small, when the complexity is
concerned in practical design, the transmitter may choose this $\bb$
which is not from iterative solving to implement the secure LA-GPC.
Furthermore, we also show the rate region derived by $\bb=[0\,\,0]$,
which is the same as treating interference as noise. This method is
still worse than the proposed method, but slightly better than the
time sharing.

In Fig. \ref{Fig_SNR_Rayleigh} and Fig. \ref{Fig_SNR_Rician} we compare the rate
regions with different transmit SNRs under both Rayleigh and Rician fading channels.
 It can be seen that the rate regions of both methods enlarge with increasing transmit
  SNR. Thus we conjecture that the rate region enlarges with increasing transmit power.
   { Note that this phenomenon is not trivial for the wiretap channel,
   not to mention the more complicated FM-GBCCM. This is because that when the transmit power increases, both the SNRs at the
legitimate receiver and the eavesdropper increase. Then the secrecy
rate may not always increase with increasing transmit power. Counter
examples are given \cite{Secrecy_Finite}\cite{Li_fading_secrecy_j}.
Indeed, the monotonically increasing property for secrecy rate (with
respect to the transmit power) was also examined in \cite[Sec
V]{Pulu_ergodic}. As described in the Remark of
\cite[P.1181]{Pulu_ergodic}, whether the monotonically increasing
properties are valid or not in some general cases are still
unknown.}

   Also note that compared to the beamformer selection \eqref{EQ_e1} and \eqref{EQ_e2}, the iterative algorithm seems to have worse performance with
    increasing transmit SNR. However, the relation reverses when the transmit SNR is small enough.
     This may imply the advantage of iterative algorithm used in low transmit SNR. In addition,
      to see the convergence speed of the proposed iterative algorithm, in Fig. \ref{Fig_iteration}
       we show the rates of $R_1$ and $R_2$ versus the number of iteration times. In this case
        we set the transmit power $P_T=10$ and the power ratio of user 1 as $\alpha=0.5$,
         and the channel statistics for Rayleigh and Rician cases are the same as those in
         Fig. \ref{Fig_Rayleigh} and Fig. \ref{Fig_Rician}, respectively. Numerical results
         show that the iterative algorithm converges within 5 steps.
         Note that although the proposed algorithm is in the form of
         fixed point iteration, the complicated formula hinders us
         to verify the convergence property by
         \cite{Ulukus_fixed_point}.

Finally, we compare the performances of our proposed iterative algorithm with whom combines the common way of finding the boundary of rate regions \cite{Tse_fadingBC} with the modified version of our iterative algorithm. That is, given $\mu\in[0,\infty)$, we solve $(\alpha^*,\mathbf{b}^*,\mathbf{K}^*_{\bm{U}_{\pi_1}},\mathbf{K}^*_{\bm{U}_{\pi_2}})$ from the optimization problem
\begin{align}\label{EQ_boundary_of_rate_region}
&\max\,\,{R_{{\pi_1}_{UB}}+\mu R_{{\pi_2}_{UB}}}\notag\\
\mbox{s.t.}~
&\mathbf{K}_{\bm{U}_{\pi_1}}\succeq 0,~\mathbf{K}_{\bm{U}_{\pi_2}}\succeq 0~\mbox{and}~ \tr{(\mathbf{K}_{\bm{U}_{\pi_1}})}\leq \alpha P_T,~\bK_{\bm{U}_{\pi_2}}\leq (1-\alpha)P_T. \end{align}
The rate region is the union of $\left(R_{{\pi_1}_{UB}}\left(\alpha^*,\mathbf{b}^*,\mathbf{K}^*_{\bm{U}_{\pi_1}},\mathbf{K}^*_{\bm{U}_{\pi_2}}\right),R_{{\pi_2}_{UB}}\left(\alpha^*,\mathbf{b}^*,\mathbf{K}^*_{\bm{U}_{\pi_1}},\mathbf{K}^*_{\bm{U}_{\pi_2}}\right)\right)$ for all $\mu\in[0,\infty)$.\\
Note that via K.K.T. we use an algorithm similar to that depicted in Fig. \ref{Fig_flow_chart_alg} to solve \eqref{EQ_boundary_of_rate_region}.
In Fig. \ref{Rayliegh_compare}, we compare both methods under fast Rayleigh fading channels with $P_T=10$. From the figure, we can find that our proposed iterative algorithm outperforms solving \eqref{EQ_boundary_of_rate_region}.

\section{Conclusion}\label{Sec_conclusion}
In this paper we considered the secure transmission over the fast fading multiple antenna Gaussian broadcast channels with confidential messages
 (FMGBC-CM), where a multiple-antenna transmitter sends independent confidential messages to two users with information theoretic secrecy
 and only the statistics of the receivers' channel state information are known at the transmitter. We first used the same marginal property of the
 FMGBC-CM to derive the conditions that the channels not degraded to the common wiretap ones. We then derived the achievable rate region for
 the FMGBC-CM by solving the channel input covariance matrices and the inflation factor. We also proposed an iterative algorithm to solve the channel input covariance matrices and the inflation factor. Due to the complicated rate region formulae in normal SNR, we provided a low SNR analysis for finding the
 asymptotic property of the channel. Numerical examples demonstrated that both users
 can achieve positive rates simultaneously under the information-theoretic secrecy requirement. Numerical examples also show the proposed beamformer selection and iterative algorithm may
 outperform each other under different conditions. 

\section{Appendix: Proof of Lemma \ref{Lemma_rate_region}}\label{Sec_App_Lemma3}
\begin{proof}
To avoid the abuse of notations, in the proof, we assume the message of User 1 is encoded first without loss of generality.
We first derive each part of (\ref{EQ_R1_GBC}) from (\ref{EQ_R1_DMC2}) and (\ref{EQ_R2_DMC2}) in the following
\begin{align}
I(\bm{V}_1;Y_2|\bm{V}_2,{\bm{H}_2})=&h(Y_2|\bm{V}_2,\!{\bm{H}_2})\!-h(Y_2|\bm{V}_1,\!\bm{V}_2,\!{\bm{H}_2})=h({\bm{H}_2}^H\bm{U}_1\!+\!Z_2)\!-h(Z_2)\!=\!E_{\bm{H}_2}\!\!\left[\log(1\!+\!{\bm{H}_2}^H\mathbf{K}_{\bm{U}_1}\!{\bm{H}_2})\right],\notag\\
I(\bm{V}_1;\bm{V}_2)=&h(\bm{V}_1)-H(\bm{V}_1|\bm{V}_2)=h(\bm{U}_1'+\mathbf{b}\bm{U}_2)-h(\bm{U}_1')=\log|\mathbf{I}+\mathbf{bK}_{\bm{U}_2}\mathbf{b}^H|,\notag\\
I(\bm{V}_1;Y_1|{\bm{H}_1})=&h({\bm{H}_1}^H(\bm{U}_1+\bm{U}_2)+Z_1)-h({\bm{H}_1}^H(\bm{U}_1+\bm{U}_2)+Z_1|\bm{U}_1'+\mathbf{b}\bm{U}_2,{\bm{H}_1})\notag\\
=&h({\bm{H}_1}^H(\bm{U}_1\!\!+\!\!\bm{U}_2)\!\!+\!\!Z_1)\!-\!h({\bm{H}_1}^H(\bm{U}_1+\bm{U}_2)-\mathbf{s}^H(\bm{U}_1'+\mathbf{b}\bm{U}_2)+Z_1|\bm{U}_1'+\mathbf{b}\bm{U}_2,{\bm{H}_1})\label{r_dpc}\\
=&h({\bm{H}_1}^H(\bm{U}_1+\bm{U}_2)+Z_1)-h({\bm{H}_1}^H(\bm{U}_1+\bm{U}_2)-\mathbf{s}^H(\bm{U}_1'+\mathbf{b}\bm{U}_2)+Z_1|{\bm{H}_1}) \label{r_dpc2}\\
=&E_{\bm{H}_1}\left[\log(1+{\bm{H}_1}^H(\mathbf{K}_{\bm{U}_1}+\mathbf{K}_{\bm{U}_2}){\bm{H}_1}) \right. \notag\\
&\left.-\log(1+{\bm{H}_1}^H(\mathbf{K}_{\bm{U}_1}+\mathbf{K}_{\bm{U}_2}){\bm{H}_1} -{\bm{H}_1}^H(\mathbf{T}_1+\mathbf{K}_{\bm{U}_2}\mathbf{b}^H)(\mathbf{I}+\mathbf{bK}_{\bm{U}_2}\mathbf{b}^H)^{-1}(\mathbf{T}_1^H+\mathbf{bK}_{\bm{U}_2}){\bm{H}_1})\right], \label{r_dpc3}
\end{align}
where in (\ref{r_dpc}), $\mathbf{s}$ is chosen so that
${\bm{H}_1}^H(\bm{U}_1+\bm{U}_2)+Z_1-\mathbf{s}^H(\bm{U}_1'+\mathbf{b}\bm{U}_2)$
is independent of $\bm{U}_1'+\mathbf{b}\bm{U}_2$. To achieve this,
the selected $\mathbf{s}$ is the MMSE estimator of
$x={\bm{H}_1}^H(\bm{U}_1+\bm{U}_2)+Z_1$ given the observation
$y=\bm{U}_1'+\mathbf{b}\bm{U}_2$, i.e.,
$\mathbf{s}^H={\bm{H}_1}^H(\mathbf{T}_1+\mathbf{K}_{\bm{U}_2}\mathbf{b}^H)(\mathbf{I}+\mathbf{bK}_{\bm{U}_2}\mathbf{b}^H)^{-1}$.
Substitute this value of $\mathbf{s}^H$ into (\ref{r_dpc2}), and
after some arrangement we can get (\ref{r_dpc3}). Combining the
above and using the fact that

$\det\left(
                                                                                     \begin{array}{cc}
                                                                                       A & B \\
                                                                                  \vspace{-0cm}     C & D \\
                                                                                     \end{array}
                                                                                   \right)
=\det(A)\det(D-CA^{-1}B)$ for nonsingular $A$, we can obtain that
\begin{align}
R_1 \leq& \left( I(\bm{V}_1;Y_1|{\bm{H}_1})-I(\bm{V}_1;Y_2|\bm{V}_2,{\bm{H}_2})-I(\bm{V}_1;\bm{V}_2) \right)^+    \notag\\
=&\left( E_{\bm{H}_1}\left[\log(1+{\bm{H}_1}^H(\mathbf{K}_{\bm{U}_1}+\mathbf{K}_{\bm{U}_2}){\bm{H}_1})\right.\right.\\
&\left.-\log(1+{\bm{H}_1}^H(\mathbf{K}_{\bm{U}_1}+\mathbf{K}_{\bm{U}_2}){\bm{H}_1}
-{\bm{H}_1}^H(\mathbf{T}_1+\mathbf{K}_{\bm{U}_2}\mathbf{b}^H)(\mathbf{I}+\mathbf{b}K_{\bm{U}_2}\mathbf{b}^H)^{-1}(\mathbf{T}_1^H+\mathbf{bK}_{\bm{U}_2}){\bm{H}_1})\right]\notag\\
&\left.-E_{\bm{H}_2}\left[\log(1+{\bm{H}_2}^H\mathbf{K}_{\bm{U}_1}{\bm{H}_2})\right]-\log|\mathbf{I}+\mathbf{bK}_{\bm{U}_2}\mathbf{b}^H| \right)^+\notag\\
=&\Bigg( E_{\bm{H}_1}[\log(1+{\bm{H}_1}^H(\mathbf{K}_{\bm{U}_1}+\mathbf{K}_{\bm{U}_2}){\bm{H}_1})]-E_{\bm{H}_2}[\log(1+{\bm{H}_2}^H\mathbf{K}_{\bm{U}_1}{\bm{H}_2})]-\notag\\
&E_{\bm{H}_1}\left[\log \left| \begin{array}{cc}{ \mathbf{I}+\mathbf{bK}_{\bm{U}_2}\mathbf{b}^H }&{ (\mathbf{T}_1^H+\mathbf{bK}_{\bm{U}_2}){\bm{H}_1} }  \\                                                                                                 { {\bm{H}_1}^H(\mathbf{T}_1+\mathbf{K}_{\bm{U}_2}\mathbf{b}^H) }&{ 1+{\bm{H}_1}^H(\mathbf{K}_{\bm{U}_1}+\mathbf{K}_{\bm{U}_2}){\bm{H}_1} }
                                                                                                  \end{array}
    \right|   \right] \Bigg)^+.   \label{r1_ub}
\end{align}

Similarly, we can get (\ref{EQ_R2_GBC}) which completes the proof of Lemma \ref{Lemma_rate_region}.
\end{proof}

\section{Appendix: Proof of Lemma \ref{lemma_optimal_input_cov_mat}}\label{Sec_App_Lemma4}
\begin{proof}
Instead of solving $\bK_{\bm{U}_{\pi_1}}$ and $\bK_{\bm{U}_{\pi_2}}$ from \eqref{EQ_R1_GBC} and \eqref{EQ_R2_GBC} directly, which may be intractable, we resort to solving the upper bound (UB) of the rate region described by Lemma \ref{Lemma_rate_region}. That is, the transmitter can use full CSIT to design the inflation factor. Then it is clear that the solution of $\bK_{\bm{U}_{\pi_1}}$ and $\bK_{\bm{U}_{\pi_2}}$ is suboptimal for \eqref{EQ_R1_GBC} and \eqref{EQ_R2_GBC}. Note also that the optimal $\bb$ for the UB is the MMSE estimator shown in \eqref{EQ_MMSE_b}.
Note that (\ref{EQ_MMSE_b}) comes from that
$\mathbf{b}=\mathbf{a}\bm{H}_{\pi_1}^H$ and $\mathbf{a}$ is the MMSE estimator of $\bm{U}_{\pi_1}'$ given the observation $y_{\pi_1}=\bm{H}_{\pi_1}^H\mathbf{T}_{\pi_1}\bm{U}_{\pi_1}'+z_{\pi_1}$.

With this choice of $\mathbf{b}$, we can eliminate the interference
perfectly by DPC. Thus we have the following rate \\region UB
\begin{align}
&(R_{\pi_1},R_{\pi_2})_{UB}\in \mbox{co}\left\{\underset{0\leq\alpha\leq 1}\bigcup\mathcal{R}(\mathbf{K}_{\bm{U}_{\pi_1}},\mathbf{K}_{\bm{U}_{\pi_2}})\right\},\notag\\
\end{align}
where $R_{\pi_1}$ and $R_{\pi_2}$ have the same forms as the RHS of \eqref{EQ_full_CSIT_R1} and \eqref{EQ_full_CSIT_R2}, respectively, except $\mathbf{K}_{\bm{U}_{\pi_1}}$ and $\mathbf{K}_{\bm{U}_{\pi_2}}$. This is because in \eqref{EQ_full_CSIT_R1} and \eqref{EQ_full_CSIT_R2}, $\mathbf{K}_{\bm{U}_{\pi_1}}$ and $\mathbf{K}_{\bm{U}_{\pi_2}}$ depend on $\bm{H}_{\pi_1}$ and $\bm{H}_{\pi_2}$, respectively, due to full CSIT. But in our model, $\mathbf{K}_{\bm{U}_{\pi_1}}$ and $\mathbf{K}_{\bm{U}_{\pi_2}}$ only depends on the statistics of the channels.

After applying Jensen's inequality to (\ref{EQ_full_CSIT_R1}), we can derive the UB of $R_1$ as
\begin{align}
R_1&\leq \left( E\left[ \log \frac{1+{\bm{H}_{\pi_1}}^H\mathbf{K}_{\bm{U}_{\pi_1}}{\bm{H}_{\pi_1}}}{1+{\bm{H}_{\pi_2}}^H\mathbf{K}_{\bm{U}_{\pi_1}}{\bm{H}_{\pi_2}}}   \right] \right)^+\overset{(a)}= \left( E\left[ \log \frac{|\mathbf{I}+{\bm{H}_{\pi_1}}{\bm{H}_{\pi_1}}^H\mathbf{K}_{\bm{U}_{\pi_1}}|}{|\mathbf{I}+{\bm{H}_{\pi_2}\bm{H}_{\pi_2}}^H\mathbf{K}_{\bm{U}_{\pi_1}}|}   \right] \right)^+ \notag\\
&\overset{(b)}\cong \left( \log \frac{|E[\mathbf{I}+{\bm{H}_{\pi_1}}{\bm{H}_{\pi_1}}^H\mathbf{K}_{\bm{U}_{\pi_1}}]|}{|E[\mathbf{I}+{\bm{H}_{\pi_2}\bm{H}_{\pi_2}}^H\mathbf{K}_{\bm{U}_{\pi_1}}]|} \right)^+ \overset{(c)}= \left(\log \frac{|\mathbf{I}+(\bK_{{\bm{H}}_{\pi_1}}+\bm{\mu}_{\pi_1}\bm{\mu}_{\pi_1}^H)\mathbf{K}_{\bm{U}_{\pi_1}}|}{|\mathbf{I}+(\bK_{{\bm{H}}_{\pi_2}}+\bm{\mu}_{\pi_2}\bm{\mu}_{\pi_2}^H)\mathbf{K}_{\bm{U}_{\pi_1}}|} \right)^+, \label{r1_liu_3}
\end{align}
where (a) follows from the Sylvester's determinant theorem
$|\mathbf{I}_m+\mathbf{A}_{m\times n}\mathbf{B}_{n\times
m}|=|\mathbf{I}_n+\mathbf{B}_{n\times m}\mathbf{A}_{m\times n}|$,
and (b) is due to the fact that the function $\log |\mathbf{X}|$ is
concave of $\mathbf{X}$. Finally, (c) is due to
$E[{\bm{H}_{\pi_1}\bm{H}_{\pi_1}}^H]=\bK_{{\bm{H}}_{\pi_1}}+\bm{\mu}_{\pi_1}\bm{\mu}_{\pi_1}^H$.
Note that the RHS of the equality (c) is the same as the secrecy
capacity of a multiple-input multiple-output multiple-antenna
eavesdropper's channel matrices with equivalent main and
eavesdropper channels as
$(\bK_{{\bm{H}}_{\pi_1}}+\bm{\mu}_{\pi_1}\bm{\mu}_{\pi_1}^H)^{1/2}$
and
$(\bK_{{\bm{H}}_{\pi_2}}+\bm{\mu}_{\pi_2}\bm{\mu}_{\pi_2}^H)^{1/2}$,
respectively. And for that channel the general optimal input
covariance matrix is unknown. Partial results can be referred to
\cite{Fakoorian_MIMO_wiretap}. Here we adopt the beamformer (rank-1)
as the input covariance matrix, i.e.,
$\mathbf{K}_{\bm{U}_{\pi_1}}=\alpha
P\mathbf{e}_{\pi_1}\mathbf{e}_{\pi_1}^H$, where
$||\mathbf{e}_{\pi_1}||^2=1$. Then  we can further rearrange
\eqref{r1_liu_3} as
\begin{align}
R_1\leq & \left( \log \frac{|\mathbf{I}+(\bK_{{\bm{H}}_{\pi_1}}+\bm{\mu}_{\pi_1}\bm{\mu}_{\pi_1}^H)\mathbf{K}_{\bm{U}_{\pi_1}}|}{|\mathbf{I}+(\bK_{{\bm{H}}_{\pi_2}}+\bm{\mu}_{\pi_2}\bm{\mu}_{\pi_2}^H)\mathbf{K}_{\bm{U}_{\pi_1}}|} \right)^+
= \left( \log \frac{|\mathbf{I}+\alpha P(\bK_{{\bm{H}}_{\pi_1}}+\bm{\mu}_{\pi_1}\bm{\mu}_{\pi_1}^H)\mathbf{e}_{\pi_1}\mathbf{e}_{\pi_1}^H|}{|\mathbf{I}+\alpha P(\bK_{{\bm{H}}_{\pi_2}}+\bm{\mu}_{\pi_2}\bm{\mu}_{\pi_2}^H)\mathbf{e}_{\pi_1}\mathbf{e}_{\pi_1}^H|} \right)^+ \notag\\
=& \left(\log \frac{\mathbf{e}_{\pi_1}^H(\mathbf{I}+\alpha P(\bK_{{\bm{H}}_{\pi_1}}+\bm{\mu}_{\pi_1}\bm{\mu}_{\pi_1}^H))\mathbf{e}_{\pi_1}}{\mathbf{e}_{\pi_1}^H(\mathbf{I}+\alpha P(\bK_{{\bm{H}}_{\pi_2}}+\bm{\mu}_{\pi_2}\bm{\mu}_{\pi_2}^H))\mathbf{e}_{\pi_1}} \right)^+. \label{r1_liu_4}
\end{align}
Since \eqref{r1_liu_4} is the Rayleigh quotient, it is known that
the optimal $\mathbf{e}_{\pi_1}$ is the eigenvector corresponding to
the maximum generalized eigenvalue of $\big(~ \mathbf{I}+\alpha
P(\mathbf{K}_{\bm{H}_{\pi_1}}+\bm{\mu}_{\pi_1}\bm{\mu}_{\pi_1}^H)
~,~ \mathbf{I}+\alpha
P(\mathbf{K}_{\bm{H}_{\pi_2}}+\bm{\mu}_{\pi_2}\bm{\mu}_{\pi_2}^H)~\big)$.

Following the similar logic, we can derive the UB of $R_2$ by
\begin{align}
R_2&\leq \left( E\left[ \log \frac{[1+{\bm{H}_{\pi_2}}^H(\mathbf{K}_{\bm{U}_{\pi_1}}+\mathbf{K}_{\bm{U}_{\pi_2}}){\bm{H}_{\pi_2}}][1+{\bm{H}_{\pi_1}}^H\mathbf{K}_{\bm{U}_{\pi_1}}{\bm{H}_{\pi_1}}]}{[1+{\bm{H}_{\pi_1}}^H(\mathbf{K}_{\bm{U}_{\pi_1}}+\mathbf{K}_{\bm{U}_{\pi_2}}){\bm{H}_{\pi_1}}][1+{\bm{H}_{\pi_2}}^H\mathbf{K}_{\bm{U}_{\pi_1}}{\bm{H}_{\pi_2}}]}                                 \right] \right)^+\notag\\
&\overset{(a)}\cong \left( \log \frac{E[1+{\bm{H}_{\pi_2}}^H(\mathbf{K}_{\bm{U}_{\pi_1}}+\mathbf{K}_{\bm{U}_{\pi_2}}){\bm{H}_{\pi_2}}] E[1+{\bm{H}_{\pi_1}}^H\mathbf{K}_{\bm{U}_{\pi_1}}{\bm{H}_{\pi_1}}]}{E[1+{\bm{H}_{\pi_1}}^H(\mathbf{K}_{\bm{U}_{\pi_1}}+\mathbf{K}_{\bm{U}_{\pi_2}}){\bm{H}_{\pi_1}}] E[1+{\bm{H}_{\pi_2}}^H\mathbf{K}_{\bm{U}_{\pi_1}}{\bm{H}_{\pi_2}}]} \right)^+\notag\\
&= \left(\log \left( \left[1+\frac{E[{\bm{H}_{\pi_2}}^H\mathbf{K}_{\bm{U}_{\pi_2}}{\bm{H}_{\pi_2}}]}{1+E[{\bm{H}_{\pi_2}}^H\mathbf{K}_{\bm{U}_{\pi_1}}{\bm{H}_{\pi_2}}]}\right]\Bigg/\left[1+\frac{E[{\bm{H}_{\pi_1}}^H\mathbf{K}_{\bm{U}_{\pi_2}}{\bm{H}_{\pi_1}}]}{1+E[{\bm{H}_{\pi_1}}^H\mathbf{K}_{\bm{U}_{\pi_1}}{\bm{H}_{\pi_1}}]}\right]\right) \right)^+ \notag\\
&\overset{(b)} = \left(\log \left(\left[1+\frac{E[{\bm{H}_{\pi_2}}^H\mathbf{K}_{\bm{U}_{\pi_2}}{\bm{H}_{\pi_2}}]}{1+\alpha P\mathbf{e}_{\pi_1}^H(\bK_{{\bm{H}}_{\pi_2}}+\bm{\mu}_{\pi_2}\bm{\mu}_{\pi_2}^H)\mathbf{e}_{\pi_1}}\right]\Bigg/\left[1+\frac{E[{\bm{H}_{\pi_1}}^H\mathbf{K}_{\bm{U}_{\pi_2}}{\bm{H}_{\pi_1}}]}{1+\alpha P\mathbf{e}_{\pi_1}^H(\bK_{{\bm{H}}_{\pi_1}}+\bm{\mu}_{\pi_1}\bm{\mu}_{\pi_1}^H)\mathbf{e}_{\pi_1}}\right]\right) \right)^+\notag \\
&= \left(\log \left(E \left[1+\frac{{\bm{H}_{\pi_2}}^H\mathbf{K}_{\bm{U}_{\pi_2}}{\bm{H}_{\pi_2}}}{1+\alpha P\mathbf{e}_{\pi_1}^H(\bK_{{\bm{H}}_{\pi_2}}+\bm{\mu}_{\pi_2}\bm{\mu}_{\pi_2}^H)\mathbf{e}_{\pi_1}}\right]\Bigg/E \left[1+\frac{{\bm{H}_{\pi_1}}^H\mathbf{K}_{\bm{U}_{\pi_2}}{\bm{H}_{\pi_1}}}{1+\alpha P\mathbf{e}_{\pi_1}^H(\bK_{{\bm{H}}_{\pi_1}}+\bm{\mu}_{\pi_1}\bm{\mu}_{\pi_1}^H)\mathbf{e}_{\pi_1}}\right]\right) \right)^+\notag\\
&\overset{(c)}= \left(\log \left(\Bigg|\mathbf{I}+\frac{(\bK_{{\bm{H}}_{\pi_2}}+\bm{\mu}_{\pi_2}\bm{\mu}_{\pi_2}^H)\mathbf{K}_{\bm{U}_{\pi_2}}}{1+\alpha P\mathbf{e}_{\pi_1}^H(\bK_{{\bm{H}}_{\pi_2}}+\bm{\mu}_{\pi_2}\bm{\mu}_{\pi_2}^H)\mathbf{e}_{\pi_1}}\Bigg|\Bigg/\Bigg|\mathbf{I}+\frac{(\bK_{{\bm{H}}_{\pi_1}}+\bm{\mu}_{\pi_1}\bm{\mu}_{\pi_1}^H)\mathbf{K}_{\bm{U}_{\pi_2}}}{1+\alpha P\mathbf{e}_{\pi_1}^H(\bK_{{\bm{H}}_{\pi_1}}+\bm{\mu}_{\pi_1}\bm{\mu}_{\pi_1}^H)\mathbf{e}_{\pi_1}}\Bigg| \right)\right)^+. \label{r2_liu_1}
\end{align}
Note that (a) is by applying the Jensen's inequality to both the numerator and denominator inside the logarithm and (b) comes from substituting  $\mathbf{K}_{\bm{U}_{\pi_1}}=\alpha P\mathbf{e}_{\pi_1}\mathbf{e}_{\pi_1}^H$. In (c) we use again the Sylvester's determinant theorem.

Again, let $\mathbf{K}_{\bm{U}_{\pi_2}}=(1-\alpha)
P\mathbf{e}_{\pi_2}\mathbf{e}_{\pi_2}^H$, with
$||\mathbf{e}_{\pi_2}||^2=1$, we have
\begin{align}
R_2\leq & \left(\log \frac{\Big|\mathbf{I}+\frac{(\bK_{{\bm{H}}_{\pi_2}}+\bm{\mu}_{\pi_2}\bm{\mu}_{\pi_2}^H)\mathbf{K}_{\bm{U}_{\pi_2}}}{1+\alpha P\mathbf{e}_{\pi_1}^H(\bK_{{\bm{H}}_{\pi_2}}+\bm{\mu}_{\pi_2}\bm{\mu}_{\pi_2}^H)\mathbf{e}_{\pi_1}}\Big|}{\Big|\mathbf{I}+\frac{(\bK_{{\bm{H}}_{\pi_1}}+\bm{\mu}_{\pi_1}\bm{\mu}_{\pi_1}^H)\mathbf{K}_{\bm{U}_{\pi_2}}}{1+\alpha P\mathbf{e}_{\pi_1}^H(\bK_{{\bm{H}}_{\pi_1}}+\bm{\mu}_{\pi_1}\bm{\mu}_{\pi_1}^H)\mathbf{e}_{\pi_1}}\Big|} \right)^+ = \left(\log \frac{\mathbf{e}_{\pi_2}^H\Big(\mathbf{I}+\frac{(1-\alpha)P(\bK_{{\bm{H}}_{\pi_2}}+\bm{\mu}_{\pi_2}\bm{\mu}_{\pi_2}^H)}{1+\alpha P\mathbf{e}_{\pi_1}^H(\bK_{{\bm{H}}_{\pi_2}}+\bm{\mu}_{\pi_2}\bm{\mu}_{\pi_2}^H)\mathbf{e}_{\pi_1}}\Big)\mathbf{e}_{\pi_2}}{\mathbf{e}_{\pi_2}^H\Big(\mathbf{I}+\frac{(1-\alpha)P(\bK_{{\bm{H}}_{\pi_1}}+\bm{\mu}_{\pi_1}\bm{\mu}_{\pi_1}^H)}{1+\alpha P\mathbf{e}_{\pi_1}^H(\bK_{{\bm{H}}_{\pi_1}}+\bm{\mu}_{\pi_1}\bm{\mu}_{\pi_1}^H)\mathbf{e}_{\pi_1}}\Big)\mathbf{e}_{\pi_2}} \right)^+
\end{align}
And the optimal $\mathbf{e}_{\pi_2}$ is the eigenvector corresponding to the maximum generalized eigenvalue of
\[
\left(~\mathbf{I}+\frac{(1-\alpha)P(\bK_{{\bm{H}}_{\pi_2}}+\bm{\mu}_{\pi_2}\bm{\mu}_{\pi_2}^H)}{1+\alpha P\mathbf{e}_{\pi_1}^H(\bK_{{\bm{H}}_{\pi_2}}+\bm{\mu}_{\pi_2}\bm{\mu}_{\pi_2}^H)\mathbf{e}_{\pi_1}} ~,~\mathbf{I}+\frac{(1-\alpha)P(\bK_{{\bm{H}}_{\pi_1}}+\bm{\mu}_{\pi_1}\bm{\mu}_{\pi_1}^H)}{1+\alpha P\mathbf{e}_{\pi_1}^H(\bK_{{\bm{H}}_{\pi_1}}+\bm{\mu}_{\pi_1}\bm{\mu}_{\pi_1}^H)\mathbf{e}_{\pi_1}}~\right).
\]
\vspace{-0.5cm}
\end{proof}

\section{Appendix: Proof of Lemma \ref{lemma_KKT}}\label{Sec_App_b}
In the proof, we first transform the unknown variables to be solved from $\bK_{\bm{U}_{\pi_k}}$ to $\bT_{\pi_k}$ by the decomposition $\mathbf{K}_{\bm{U}_{\pi_1}}=\mathbf{T}_{\pi_1}\mathbf{T}_{\pi_1}^H$ and $\mathbf{K}_{\bm{U}_{\pi_2}}=\mathbf{T}_{\pi_2}\mathbf{T}_{\pi_2}^H$, such that the Lagrangians can be simplified. And the resulting Lagrangians are
\[
L_1(\mathbf{b},\mathbf{T}_{\pi_1})\triangleq R_{\pi_1}^{UB}+\lambda_1 \left(\tr{(\mathbf{K}_{\bm{U}_{\pi_1}})}-\alpha P_T \right),\,\,L_2(\mathbf{T}_{\pi_2})\triangleq R_{\pi_2}^{UB}+\lambda_2 \left(\tr{(\mathbf{K}_{\bm{U}_{\pi_2}})}-(1-\alpha)P_T \right),\]
where $\lambda_1$ and $\lambda_2$ are the Lagrangian multipliers. Note that with the above decompositions, the conditions that $\mathbf{K}_{\bm{U}_{\pi_1}}\succeq 0,\,\mathbf{K}_{\bm{U}_{\pi_2}}\succeq 0$ are automatically satisfied.

Next, we calculate some derivatives that would be used for the optimization problem. For $\partial L_1/\partial \mathbf{b}=0$, we get $
\partial E_{\bm{H}_{\pi_1}}[\log|\mathbf{M}|]/\partial \mathbf{b}=0, $
where $\mathbf{M}$ is defined in \eqref{EQ_def_M}, and
\begin{align}
d~E_{\bm{H}_{\pi_1}}[\log|\mathbf{M}|]
\overset{(a)}=& E_{\bm{H}_{\pi_1}}[\tr{(\mathbf{M}^{-1}d\mathbf{M})}]\overset{(b)}=E_{\bm{H}_{\pi_1}}\!\!\!\left[\tr{\left(\!\!\!\mathbf{M}^{-1}\!\!\!\left[
                                              \begin{array}{cc}
                                                (d\mathbf{b})\mathbf{K}_{\bm{U}_{\pi_2}}\mathbf{b}^H\!\!+\!\!\mathbf{b}\mathbf{K}_{\bm{U}_{\pi_2}}\!\!(d\mathbf{b})^H \!\!&\!\! (d\mathbf{b})\mathbf{K}_{\bm{U}_{\pi_2}}\!\!\bm{H}_{\pi_1} \\
                                                \bm{H}_{\pi_1}^H\mathbf{K}_{\bm{U}_{\pi_2}}(d\mathbf{b})^H \!\!&\!\! \mathbf{0} \\
                                              \end{array}\!\!\!\right]\!\!
\right)}\!\!\right] \notag\\
\overset{(c)}=&E_{\bm{H}_{\pi_1}}\left[\tr\left(\mathbf{M}^{-1}\left(\left[
                                              \begin{array}{cc}
                                                (d\mathbf{b})\mathbf{K}_{\bm{U}_{\pi_2}}\mathbf{b}^H & (d\mathbf{b})\mathbf{K}_{\bm{U}_{\pi_2}}\bm{H}_{\pi_1} \\
                                                \mathbf{0} & \mathbf{0} \\
                                              \end{array}\right]+\left[
                                              \begin{array}{cc}
                                                \mathbf{b}\mathbf{K}_{\bm{U}_{\pi_2}}(d\mathbf{b})^H& \mathbf{0} \\
                                                \bm{H}_{\pi_1}^H\mathbf{K}_{\bm{U}_{\pi_2}}(d\mathbf{b})^H  & \mathbf{0} \\
                                              \end{array}\right]\right)
\right)\right], \notag
\end{align}
where (a) is from $d~\log|\mathbf{X}|=\tr{(\mathbf{X}^{-1}d\mathbf{X})}$, and $\partial\log|\mathbf{X}|/\partial\mathbf{X}=(\mathbf{X}^T)^{-1}$, (b) comes from some derivatives calculation \cite{Vaze_09}\cite{Vaze_cite} and (c) is from the facts that $\tr{(\bA+\bB)}=\tr{(\bA)}+\tr{(\bB)}$.

Therefore, we have
\[\frac{\partial E_{\bm{H}_{\pi_1}}[\log|\bm{M}|]}{\partial \mathbf{b}}=E_{\bm{H}_{\pi_1}}\left[\left[
                                                                                    \begin{array}{cc}
                                                                                      \mathbf{I} & \mathbf{0} \\
                                                                                    \end{array}
                                                                                  \right]
\mathbf{M}^{-H}\left[
                                                                                                                           \begin{array}{c}
                                                                                                                             \mathbf{b} \\
                                                                                                                             \bm{H}_{\pi_1}^H \\
                                                                                                                           \end{array}
                                                                                                                         \right]\mathbf{K}_{\bm{U}_{\pi_2}} \right]=0.
                                                                                                                         \]
After some manipulations, it can be proved that without loss of generality, the solution of $\mathbf b$ satisfies the form described in Lemma \ref{lemma_KKT}, even when $|\mathbf{K}_{\bm{U}_{\pi_2}}|=0$.\\
To compute $\partial L_1/\partial \mathbf{T}_{\pi_1}=0$ and $\partial L_2/\partial \mathbf{T}_{\pi_2}=0$, with similar steps described above, we can have the final results.

\section{Appendix: Proof of Lemma \ref{lemma_low_snr}}\label{sec_app_lemma6}
\begin{proof}
We can easily observe that the function including $\mathbf{b}$ in
our rate formulae, i.e., the second term on the RHS of
\eqref{EQ_Delta}, has the same form as in \cite[(1)]{Vaze_09}. Then
from \cite[Sec. IV-B]{Vaze_09}, we know that when $ P\rightarrow 0
$, the optimal inflation factor in our problem is
$\mathbf{b}_{optimal}=\mathbf{0}$. That is, treating interference as
noise directly. Substituting $\mathbf{b}=\mathbf{0}$ into Lemma
\ref{Lemma_rate_region}, the achievable rate-region becomes

\begin{align}
R_j\!\! \leq\!\! \!\!\left( \!\!\! E_{\bm{H}_j}[\log(\!\!1\!\!+\!\!{\bm{H}^H_j}\!\!(\!\mathbf{K}_{\bm{U}_j}\!\!+\!\!\mathbf{K}_{\bm{U}_{j^{c}}}){\bm{H}_j})]\!\!-\!\!E_{\bm{H}_{j^{c}}}[\log(\!\!1\!\!+\!\!{\bm{H}^H_{j^{c}}}\mathbf{K}_{\bm{U}_j}{\bm{H}_{j^{c}}}\!)] \!\!\!-E_{\bm{H}_j}\!\!\!\!\left[\!\log \!\left| \!\!\!\!\begin{array}{cc}
                                                                                                   { \mathbf{I} }&{ \mathbf{T}_j^H\bm{H}_j }  \\
                                                                                                   { \bm{H}_j^H\mathbf{T}_j }&{ \!\!\!\!1\!\!+\!\!{\bm{H}^H_j}\!\!(\mathbf{K}_{\bm{U}_j}\!\!\!+\!\!\mathbf{K}_{\bm{U}_{j^{c}}}){\bm{H}_j} }
                                                                                                  \end{array}
   \!\!\!\! \right|   \right]        \!\!\! \right)^+\!\!\!. \label{rate_low_snr}
\end{align}
Besides, from block matrix determinant, we can get
\begin{align}
\left| \begin{array}{cc}
                                                                                                   { \mathbf{I} }&{ \mathbf{T}_1^H\bm{H}_1 }  \\
                                                                                                   { \bm{H}_1^H\mathbf{T}_1 }&{ 1+{\bm{H}_1}^H(\mathbf{K}_{\bm{U}_1}+\mathbf{K}_{\bm{U}_2}){\bm{H}_1} }
                                                                                                  \end{array}
    \right|
=1+\bm{H}_1^H\mathbf{K}_{\bm{U}_2}\bm{H}_1.
\end{align}

Let $\mathbf{K}_{\bm{U}_1}=\alpha P_T \mathbf{K}_1$ and $\mathbf{K}_{\bm{U}_2}=(1-\alpha) P_T \mathbf{K}_2$, where we define $\tr{(\mathbf{K}_1)}=\tr{(\mathbf{K}_2)}=1$. Then with $j=1$, (\ref{rate_low_snr}) becomes
\begin{align}
R_1 \leq & \left(  E_{\bm{H}_1}[\log(1+P_T {\bm{H}_1}^H(\alpha \mathbf{K}_1+(1-\alpha)\mathbf{K}_2){\bm{H}_1})]-E_{\bm{H}_2}[\log(1+P_T{\bm{H}_2}^H(\alpha \mathbf{K}_1)\bm{H}_2)]\right.\notag \\
&-E_{\bm{H}_1}[\log(1+P_T {\bm{H}_1}^H(1-\alpha) \mathbf{K}_2\bm{H}_1) ] \left. \right)^+\label{low_1}\\
\cong & \Big( \frac{P_T}{\ln 2}\left( \right. E_{\bm{H}_1}[\bm{H}_1^H(\alpha \mathbf{K}_1+(1-\alpha)\mathbf{K}_2)\bm{H}_1]-E_{\bm{H}_2}[{\bm{H}_2}^H(\alpha \mathbf{K}_1)\bm{H}_2]  -E_{\bm{H}_1}[{\bm{H}_1}^H(1-\alpha) \mathbf{K}_2\bm{H}_1]  \left. \right) \Big)^+ \label{low_2}\\
=& \left( \frac{P_T}{\ln 2}\left( E_{\bm{H}_1}[\bm{H}_1^H(\alpha \mathbf{K}_1)\bm{H}_1]-E_{\bm{H}_2}[\bm{H}_2^H(\alpha \mathbf{K}_1)\bm{H}_2]   \right) \right)^+,\label{low_3}
\end{align}

where (\ref{low_2}) utilizes $\lim_{x\rightarrow 0} \ln
(1+ax)=ax+O(x)$.

Similarly, we can derive $R_2$ as
\begin{align}
R_2 \cong \left( \frac{P_T}{\ln 2}\left( E_{\bm{H}_2}[\bm{H}_2^H((1-\alpha) \mathbf{K}_2)\bm{H}_2]-E_{\bm{H}_1}[\bm{H}_1^H((1-\alpha) \mathbf{K}_2)\bm{H}_1]   \right) \right)^+. \label{low_4}
\end{align}

Since $R_1$ and $R_2$ have similar structures, in the following we only derive $R_1$, and the results can be easily extended to $R_2$. Let $\bm{H}_1=\mathbf{K}_{\bm{H}_1}^{1/2}\bm{\psi}_1$ and $\bm{H}_2=\mathbf{K}_{\bm{H}_2}^{1/2}\bm{\psi}_2$, where $\bm{\psi}_1\sim CN(\mathbf{0},\mathbf{I}),\bm{\psi}_2 \sim CN(\mathbf{0},\mathbf{I})$. Then (\ref{low_3}) can be written as
\begin{align}
R_1 \cong & \left(\frac{\alpha P_T}{\ln 2} \left( E_{\bm{\psi}_1,\bm{\psi}_2}\left[\bm{\psi}_1^H\mathbf{K}_{\bm{H}_1}^{1/2}\mathbf{K}_1\mathbf{K}_{\bm{H}_1}^{1/2} \bm{\psi}_1-\bm{\psi}_2^H\mathbf{K}_{\bm{H}_2}^{1/2}\mathbf{K}_1\mathbf{K}_{\bm{H}_2}^{1/2} \bm{\psi}_2\right]  \right) \right)^+\notag\\
=& \left(\frac{\alpha P_T}{\ln 2}\left(\tr\left(\mathbf{K}_{\bm{H}_1}^{1/2}\mathbf{K}_1\mathbf{K}_{\bm{H}_1}^{1/2}\right)-\tr\left(\mathbf{K}_{\bm{H}_2}^{1/2}\mathbf{K}_1\mathbf{K}_{\bm{H}_2}^{1/2}\right)\right) \right)^+\label{low_r1_3}\\
=& \left(\frac{\alpha P_T}{\ln 2}\tr((\mathbf{K}_{\bm{H}_1}-\mathbf{K}_{\bm{H}_2})\mathbf{K}_1) \right)^+\label{low_r1_4}\\
\leq & \left( \frac{\alpha P_T}{\ln 2}\lambda_{max}(\mathbf{K}_{\bm{H}_1}-\mathbf{K}_{\bm{H}_2}) \right)^+\label{low_r1_5},
\end{align}
where (\ref{low_r1_3}) is uses the properties $\tr{(\mathbf{A}\mathbf{B})}=\tr{(\mathbf{B}\mathbf{A})}$ and $\bm{\psi}_1$ and
 $\bm{\psi}_2$ are i.i.d. with covariance matrix $\bI$; (\ref{low_r1_4}) uses the properties $\tr{(\mathbf{A}\mathbf{B})}=\tr{(\mathbf{B}\mathbf{A})}$
  again. Finally, (\ref{low_r1_5})  comes from the result of Lemma \ref{math_low} which completes the proof.
\end{proof}

\section{Appendix: Proof of Lemma \ref{math_low}}\label{sec_app_lemma7}
\begin{proof}
Let the eigen-decompositions
$\mathbf{A}=\mathbf{U}_{\mathbf{A}}\Lambda_{\mathbf{A}}\mathbf{U}_{\mathbf{A}}^H$
and
$\mathbf{M}=\mathbf{U}_{\mathbf{M}}\Lambda_{\mathbf{M}}\mathbf{U}_{\mathbf{M}}^H$,
where $\mathbf{U}_{\mathbf{A}}$ and $\mathbf{U}_{\mathbf{M}}$ are
unitary matrices and $\Lambda_{\mathbf{A}}=\diag(a_1,a_2,..,a_n)$
and $\Lambda_{\mathbf{M}}=\diag(m_1,m_2,...,m_n)$ with $a_1\geq
a_2\geq ...\geq a_n$, $m_1\geq m_2\geq ...\geq m_n\geq 0$ and
$\sum_i{m_i \leq 1}$. Then
\[
\tr(\mathbf{AM})=\tr(\mathbf{U}_{\mathbf{A}}\bm\Lambda_{\mathbf{A}}\mathbf{U}_{\mathbf{A}}^H\mathbf{U}_{\mathbf{M}}\bm\Lambda_{\mathbf{M}}\mathbf{U}_{\mathbf{M}}^H)=\tr(\mathbf{U}_{\mathbf{M}}^H\mathbf{U}_{\mathbf{A}}\bm\Lambda_{\mathbf{A}}\mathbf{U}_{\mathbf{A}}^H\mathbf{U}_{\mathbf{M}}\bm\Lambda_{\mathbf{M}})=\tr(\mathbf{U}\bm\Lambda_{\mathbf{A}}\mathbf{U}^H\bm\Lambda_{\mathbf{M}})=\tr(\mathbf{S}\bm\Lambda_{\mathbf{M}}),\]
where
$\mathbf{U}\triangleq\mathbf{U}_{\mathbf{A}}^H\mathbf{U}_{\mathbf{M}}$
is an unitary matrix, and thus
$\mathbf{S}=\mathbf{U}\bm\Lambda_{\mathbf{A}}\mathbf{U}^H$ is
Hermitian. Let $s_i$ be the $i$-th diagonal element. Hence,
$\tr(\mathbf{AM})=\tr(\mathbf{S}\bm\Lambda_{\mathbf{M}})=\sum_is_im_i\leq
\max_i(s_i)$. The equality holds if $m_k=1$, where $k=\arg
\max_{i}{s_i}$. Thus the optimal $\mathbf{M}$ is unit rank.

In the following, we solve the eigenvector of the optimal
$\mathbf{M}$. By the Schur-Horn Theorem \cite[Chapter
9]{Horn_matrix_analysis} and recall that
$\mathbf{S}=\mathbf{U}\bm\Lambda_{\mathbf{A}}\mathbf{U}^H$, we know
$\max_i~s_i\leq \max_i~a_i$. As a result, we get that
$\mathbf{S}=\mathbf{U}\bm\Lambda_{\mathbf{A}}\mathbf{U}^H=\bm\Lambda_{\mathbf{A}}$
is optimal. This happens when
$\mathbf{U}=\mathbf{U}_{\mathbf{A}}^H\mathbf{U}_{\mathbf{M}}=\mathbf{I}$,
that is, $\mathbf{U}_{\mathbf{M}}=\mathbf{U}_{\mathbf{A}}$.

From the above, we know that the eigenvector of the optimal
$\mathbf{M}$ is the one corresponding to the maximum eigenvalue of
${\mathbf{A}}$.
\end{proof}
\bibliographystyle{IEEEtran}

\renewcommand{\baselinestretch}{1}
\bibliography{IEEEabrv,SecrecyPs2}

\vspace{1cm}
\begin{figure}[hd]
\centering \epsfig{file=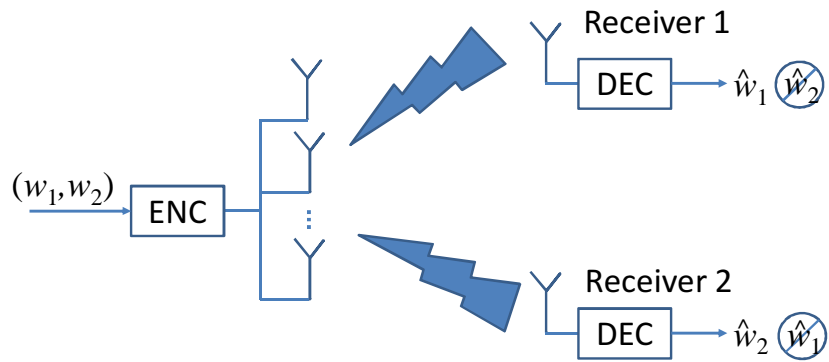, width=0.6\textwidth}
\caption{The system model of FMGBC-CM.} \label{Fig_sys}
\end{figure}

\begin{figure}[hb]
\centering \epsfig{file=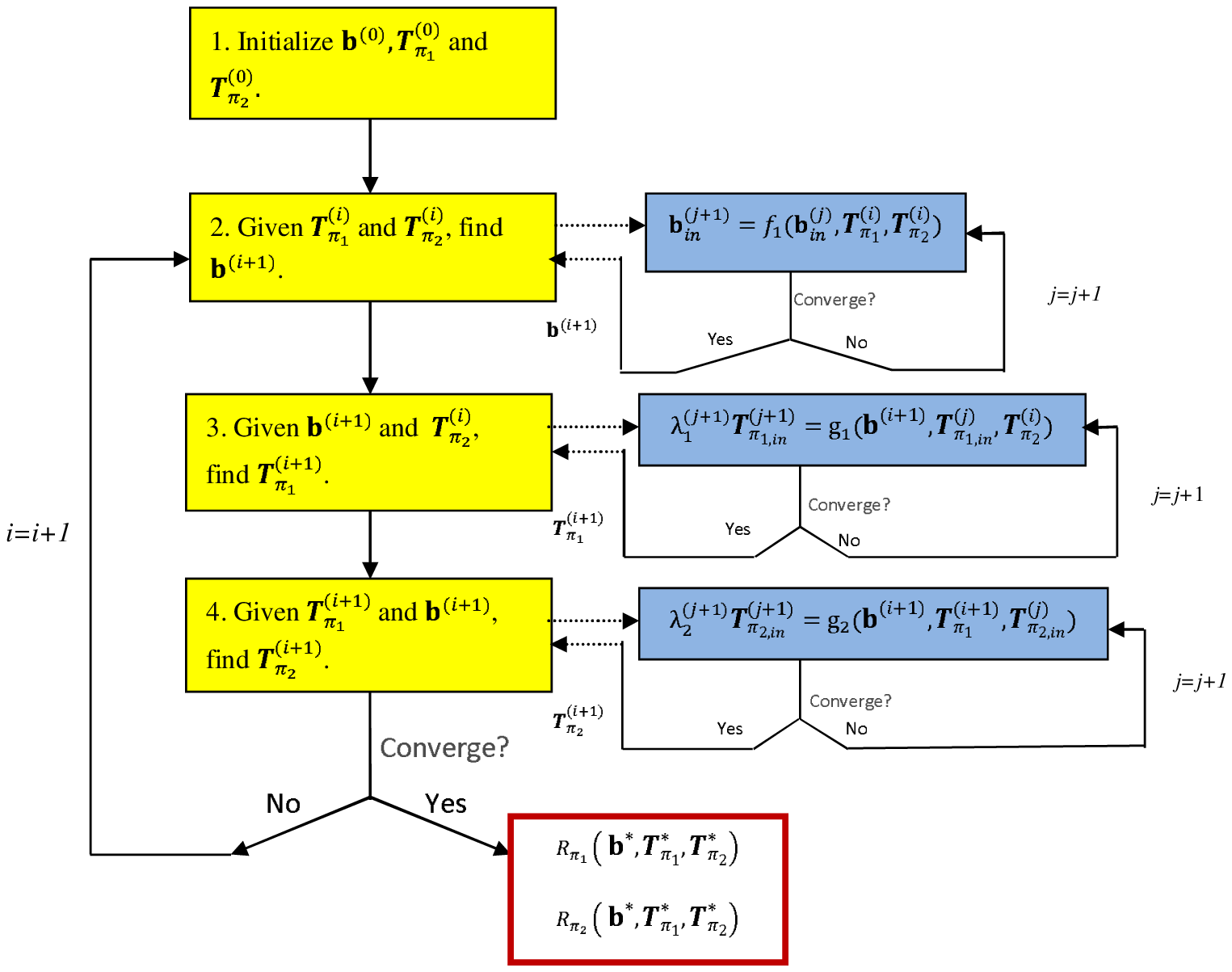 , width=0.7\textwidth}
\caption{The flow chart of the proposed algorithm.} \label{Fig_flow_chart_alg}
\end{figure}

\begin{figure}
\centering \epsfig{file=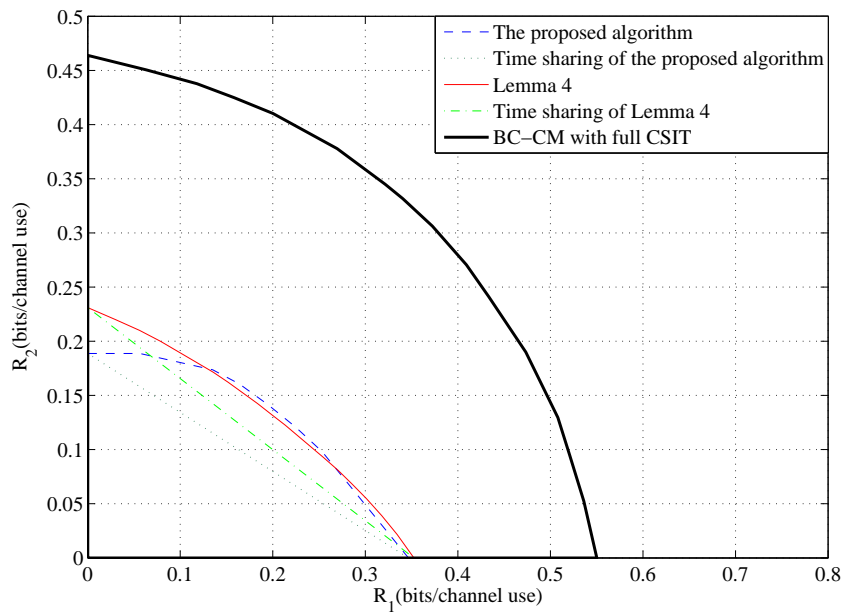 , width=0.7\textwidth}
\caption{The comparison of rate regions under fast Rayleigh fading channel with full and statistical CSIT.} \label{Fig_Rayleigh}
\end{figure}

\begin{figure}
\centering \epsfig{file=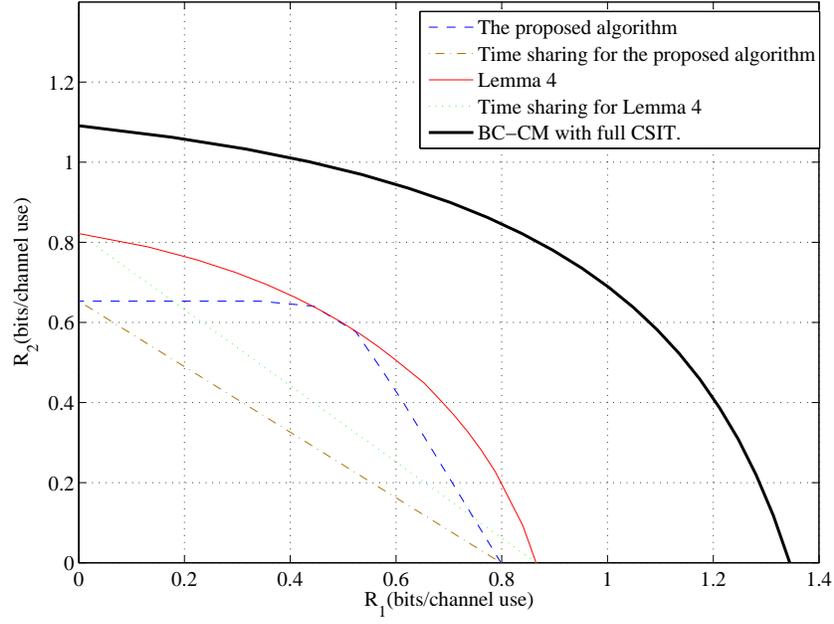 , width=0.7\textwidth}
\caption{The comparison of rate regions under fast Rician fading channel with full and statistical CSIT.} \label{Fig_Rician}
\end{figure}

\begin{figure}
\centering \epsfig{file=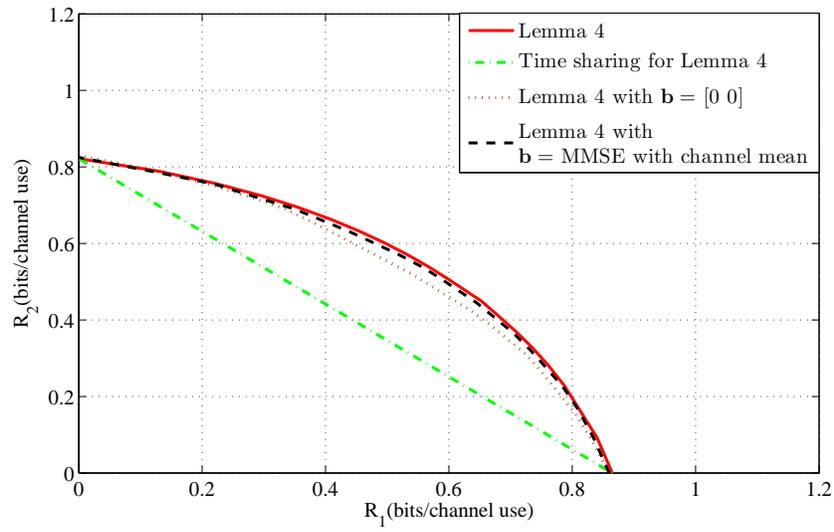 , width=0.7\textwidth}
\caption{The effect of  different choices of $\mathbf{b}$ for the proposed transmission scheme under fast Rician fading channel and statistical CSIT.} \label{Fig_Rician_b}
\end{figure}

\begin{figure}
\centering \epsfig{file=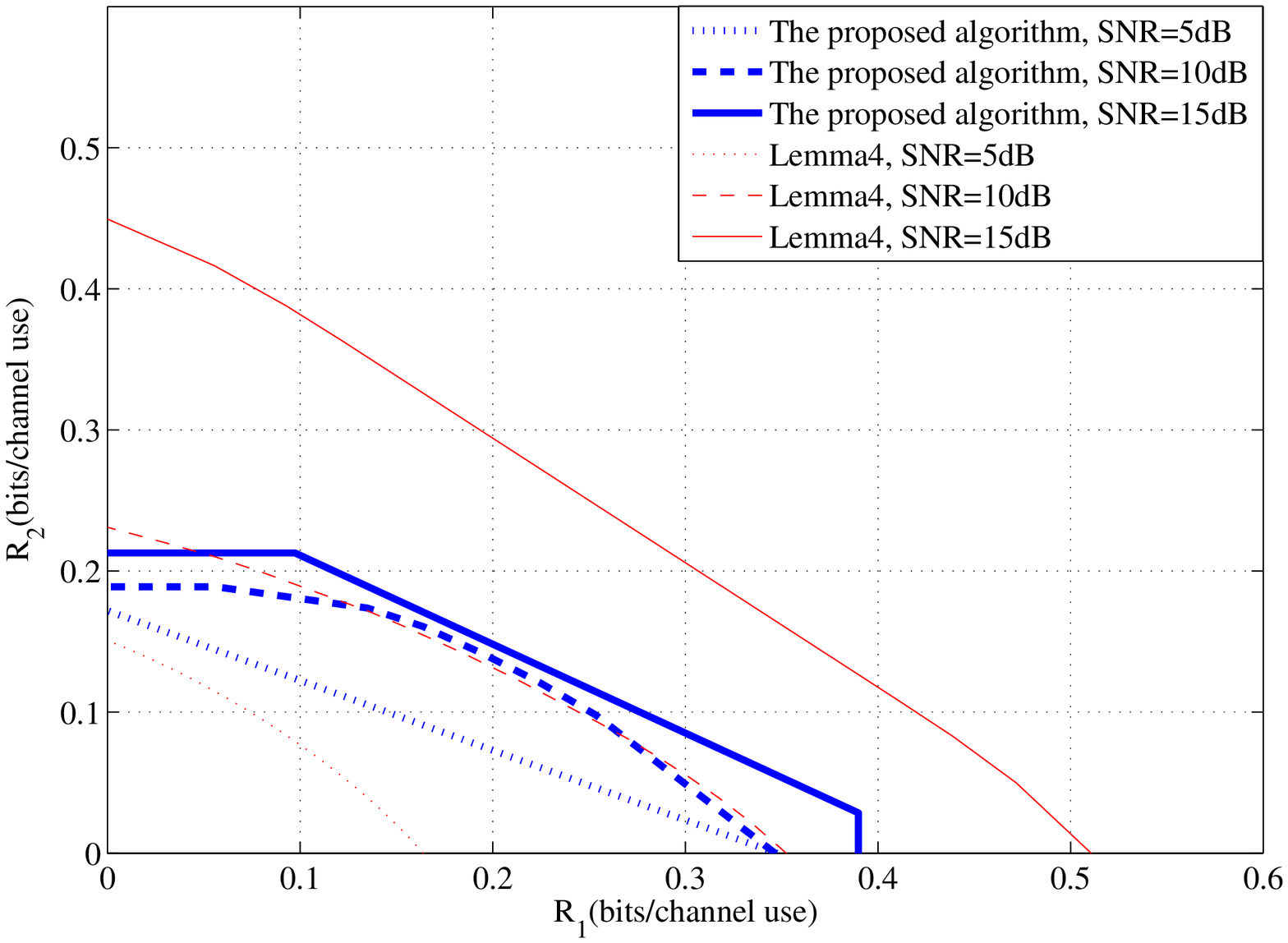 , width=0.7\textwidth}
\caption{The comparison of rate regions under fast Rayleigh fading channel with statistical CSIT and different transmit SNRs.} \label{Fig_SNR_Rayleigh}
\end{figure}

\begin{figure}
\centering \epsfig{file=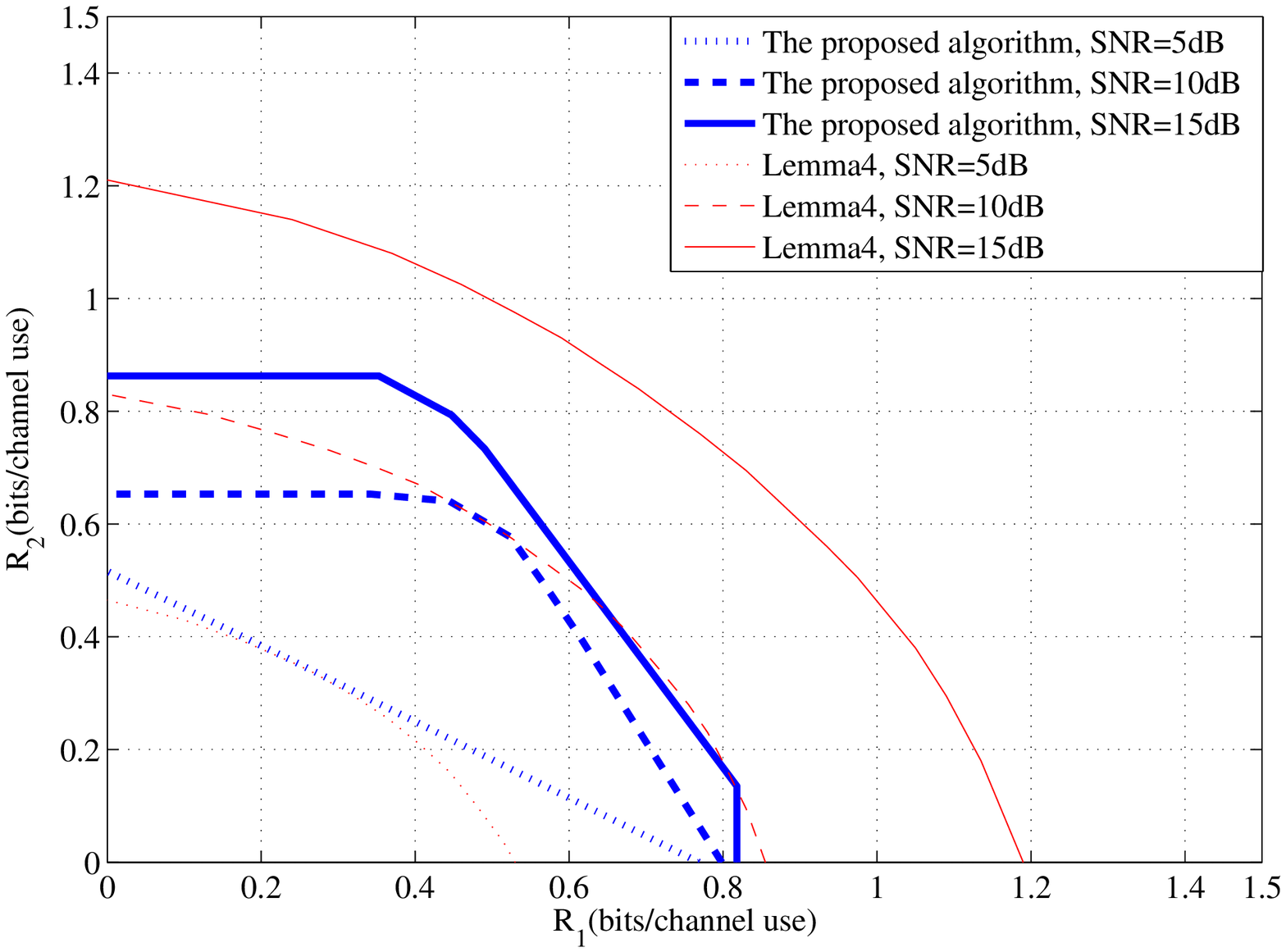 , width=0.7\textwidth}
\caption{The comparison of rate regions under fast Rician fading channel with statistical CSIT and different transmit SNRs.} \label{Fig_SNR_Rician}
\end{figure}

\begin{figure}
\centering \epsfig{file=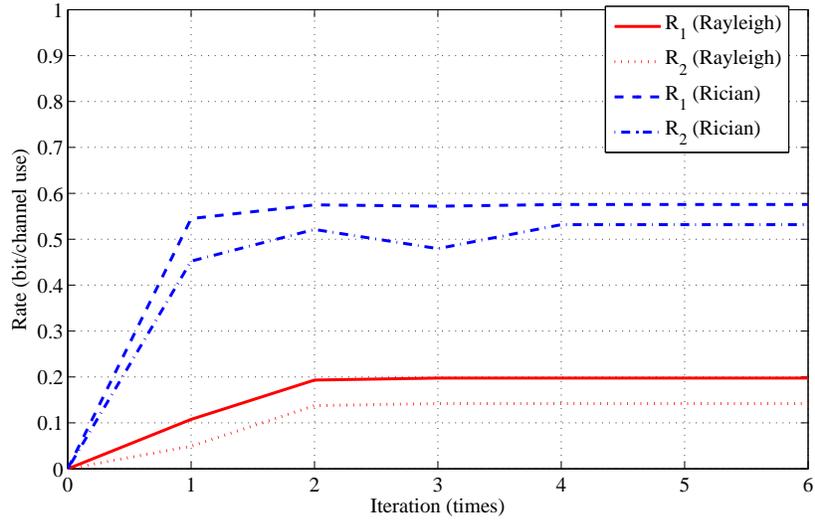 , width=0.7\textwidth}
\caption{The convergence of the proposed algorithm.}
\label{Fig_iteration}
\end{figure}

\begin{figure}
\centering \epsfig{file=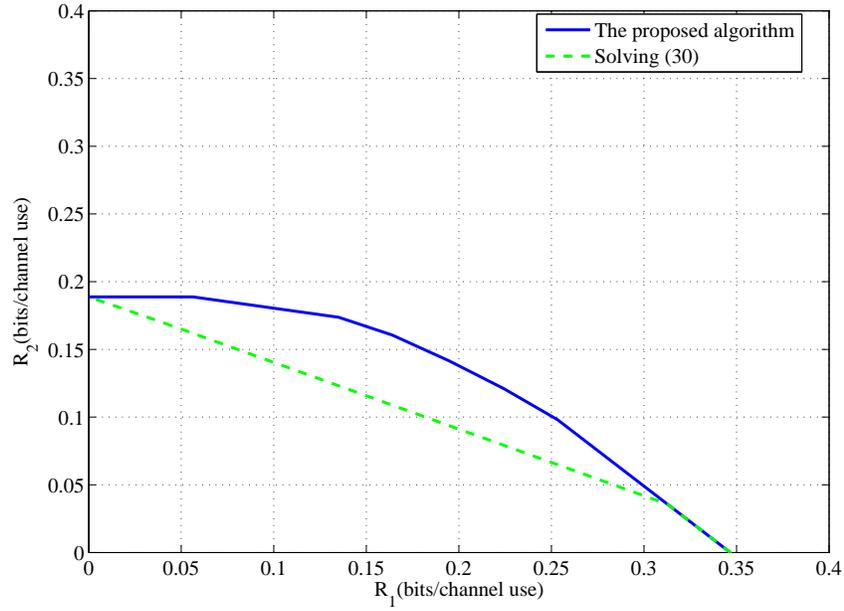 ,
width=0.7\textwidth} \caption{The comparison of the proposed
algorithm and solving (30).} \label{Rayliegh_compare}
\end{figure}

\end{document}